\documentclass[11pt,a4,reqno]{amsart}
\usepackage[square, comma, sort&compress]{natbib}
\usepackage{bbm}
\usepackage{amsmath}
\usepackage{amsfonts}
\usepackage{amssymb}
\usepackage{latexsym}
\usepackage{graphicx}
\usepackage[utf8]{inputenc}
\usepackage{mathrsfs}
\usepackage{epstopdf}
\usepackage{hyperref}

\numberwithin{equation}{section}
\newtheorem{proposition}{Proposition}[section]

\newtheorem{theorem}[proposition]{Theorem}

\newtheorem{lemma}[proposition]{Lemma}

\newcommand{\bx}{\mathbf{x}}

\newcommand{\bc}{\mathbf{c}}

\newcommand\R{\mathbb{R}}

\newcommand{\udots}{\mbox{\reflectbox{$\ddots$}}}
\newcommand{\ds}{\displaystyle}

\newcommand{\rf}[1]{(\ref{#1})}
\newcommand{\bea}{\begin{eqnarray}}
\newcommand{\eea}{\end{eqnarray}}

\def\void{}
\def\labelmark{}

\newenvironment{formula}[1]{\def\labelname{#1}
\ifx\void\labelname\def\junk{\begin{displaymath}}
\else\def\junk{\begin{equation}\label{\labelname}}\fi\junk}%
{\ifx\void\labelname\def\junk{\end{displaymath}}
\else\def\junk{\end{equation}}\fi\junk\labelmark\def\labelname{}}

{\ifx\void\labelname\def\junk{\end{array}\end{displaymath}}
\else\def\junk{\end{array}\right.\end{equation}}
\fi\junk\labelmark\def\labelname{}\def\junk{}
}

\newcommand{\beq}{\begin{formula}}
\newcommand{\eeq}{\end{formula}}
\newcommand{\beqv}{\begin{formula}{}}

\begin{document}
\title{Random tree growth by vertex splitting}
\author{F. David,  W.M.B. Dukes, T. Jonsson, S.\"O. Stef\'ansson.}

\address[Dukes\;,\;Jonsson\;,\;Stef\'ansson]{Science Institute, University of Iceland, 
	Dunhaga 3, 107 Reykjav\'ik, Iceland.}
\email{dukes/thjons@raunvis.hi.is\;,\;siguste@hi.is}

\address[David]{Institut de Physique ThÃ©orique,
CNRS, URA 2306, F-91191 Gif-sur-Yvette, France
CEA, IPhT, F-91191 Gif-sur-Yvette, France}
\email{francois.david@cea.fr}
\maketitle

\begin{center}
\today
\end{center}

\begin{abstract}
We study a model of growing planar tree graphs where in each time step we
separate the tree into two components by splitting a vertex and then
connect the two pieces by inserting a new link between the daughter
vertices. This model generalises the preferential attachment model and Ford's $\alpha$--model for phylogenetic trees. We develop a mean
field theory for the vertex degree distribution, prove that the mean
field theory is exact in some special cases and check that it agrees 
with numerical simulations in general.  We calculate various
correlation functions and show that the intrinsic Hausdorff dimension can vary
from one to infinity, depending on the parameters of the model.
\end{abstract}

\tableofcontents

\section{Introduction}
\subsection{Motivation.}
\label{ssMotivation}
Random trees arise in many branches of science ranging from
mathematics through computer science, physics and chemistry to 
biology and sociology.
Trees are the simplest generalisations of random walks and are used to
model various relationships like family trees \cite{gw} or phylogenetic 
trees \cite{kimmelaxelrod,semplesteel} and evolving populations.
They also model physical objects that look like trees, e.g. branched
polymers, and are used to encode information regarding the
secondary structure of macromolecules, RNA folding being one prominent 
example \cite{francoisetal}. 
Information about more complicated geometrical objects
like planar maps and random surfaces which arise in quantum gravity \cite{adj} can sometimes
be encoded in labelled random trees \cite{schaeffer}. 
Trees appear also in relation with fragmentation and coagulation processes \cite{Bertoin}.

There are two principal classes of models that have been used to model
trees.  The first one is equilibrium statistical mechanics, where trees
are assigned a certain weight factor (Boltzmann factor) which induces
a probability measure on the class of trees under consideration.The weight factor depends on the local properties of the trees.
The second class is growing trees where one starts with a simple tree
and then continues to add vertices and links according to some (in general stochastic) growth
rules.  After a fixed time $t$ the growth rules induce a probability
measure on the trees that can arise after time $t$.  Time can be taken
to be discrete or continuous.
The trees themselves are in general  discrete and are therefore tree graphs in the usual sense, but there exist also models of continuous random trees \cite{Aldouscont}.

The first class of models is the one that appears most
frequently in physics whereas the second class is the natural one to
use to analyse, for example growing family trees,  citation networks, and stochastic processes.  
It is well-known that in certain cases these two approaches to studying random
trees are equivalent.  
For example, the so called {\it generic random  trees} that appear in the study of gravity are equivalent to a critical
Galton-Watson process \cite{djw}.  So called {\it causal trees}  are believed to be
equivalent to certain classes of growing trees \cite{bialasII}.

In this paper we introduce and study a general model of growing trees, where growth occurs by random vertex splitting, with very general rules. 
This model includes as special limit cases several models of growing trees already studied in the physics and the mathematics literature: the well known preferential attachement model (which models for example the growth
of branched polymers out of a soup of monomers or the growth of the internet \cite{AlbertBara,MASS_DIST}); the  $\alpha$-model  \cite{alphamodel} which is a  stochastic model of cladograms (binary leaf-labelled trees studied in relation with phylogenetic trees and biological systematics) and its extensions \cite{ChenFordWinkel}. It is also related to a tree growth model that arises in the theory of RNA secondary structures \cite{francoisetal}. 
These discrete time tree growth models are in fact closely related to the continuous time tree growth models which describe e.g. self similar fragmentation processes and coalescence processes (see \cite{Bertoin} for general reference, and \cite{PitmanWinkel,HassMiermontPitmanWinkel} for recent works), and seem to be related to sequential packing processes.
Finally, back on the physics side, let us mention that growth process by vertex splitting, together with vertex merging processes, is important in Monte Carlo simulations of discrete Euclidean and Lorentzian quantum gravity models (see e.g. \cite{MCambjorn}).

Our main motivation is to develop general tools to study the properties of models of random  tree growth. 
In particular we are motivated by the issues of unification and of universality:  Is there a general tree growth process which can encompass the different models which are known at the moment? How many different universality classes, i.e.~continuous tree models with different scaling properties (exponents and correlation functions), exist in this framework?
The results presented here are a first step in this direction.

\subsection{Description of the model.}
\label{ssModelShort}
Let us first describe our model informally, a
precise definition will be given in the next section.  
Suppose that at
an initial time $t_0$ we are given a tree which will always be assumed
to be finite, rooted and planar (i.e. at each vertex the attached links
are
ordered cyclically). This assumption of planarity is convenient for the description of the model and for the proofs, but is 
not essential. In order to evolve the tree we choose a vertex
at random with a relative probability $w_i$ which only depends on 
the order $i$ of the vertex.  We split the chosen vertex into
two vertices of order $j$ and $k$ which are linked to each other and 
attach the links of the original vertex of order $i$ 
to the two new vertices so 
$j+k=i+2$ and the planarity is preserved in this process, see Fig.~\ref{F_split_0}.
\begin{figure} [!h]
\centerline{\scalebox{0.8}{\includegraphics{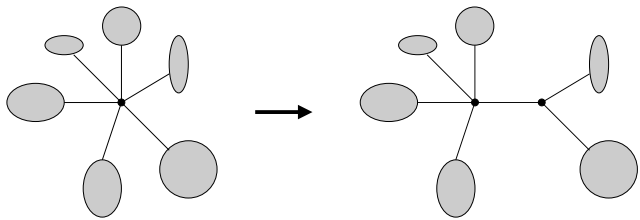}}}
\caption{The vertex splitting process.} \label{F_split_0}
\end{figure}
We sometimes let the relative probability of this splitting, $w_{j,k}$, 
depend on $j$ and $k$ but
sometimes we take the uniform distribution on the $i(i+1)/2$
possibilities.   The parameters $w_i$ and $w_{j,k}$ define the model.  
We shall consider both the case when there is
a maximal vertex degree $d_{\mathrm{max}}$, i.e., $w_{j,k}=0$ if either $j$ or $k$ exceeds
the  cutoff $d_{\mathrm{max}}$, and the general case, but most of our explicit results will be for a finite $d_{\mathrm{max}}$.

This growth process 
becomes the preferential
attachment growth model \cite{AlbertBara,MASS_DIST} when we take $w_{j,k}=0$ unless $j$ or $k$ is
equal to 1. As already mentioned, similar processes of random vertex splitting and random vertex merging are used as ergodic moves in Monte Carlo simulations of triangulations in quantum gravity, see Fig.~\ref{f:MC}.
\begin{figure} [!h]
\centerline{\scalebox{0.8}{\includegraphics{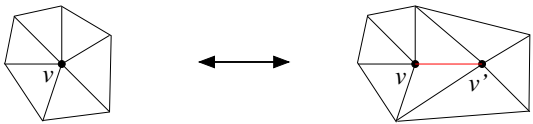}}}
\caption{Vertex splitting (left to right) and vertex merging (right to left) in a triangulation.} \label{f:MC}
\end{figure}
The connection with the models of \cite{alphamodel,ChenFordWinkel} will be discussed later.

We are interested in the structure of the trees after a large number
of steps in the growth process.  The main quantities we study are the
distribution of the order of vertices, the correlation between the
degrees of neighbouring vertices and the intrinsic Hausdorff dimension of the
trees.  In order to get at the Hausdorff dimension we shall define and
analyse certain general correlation functions which describe the subtree
structure of the random tree.

The calculations are simpler when the splitting weight $w_j$ depends linearly on $j$ ($w_j= a j+b$) because then the probability
$p_j$ of splitting a vertex of degree $j$ in a tree $T$ is given by
$p_j={w_j / F(|T|)}$
where $F$ is a function only of the number of
vertices $|T|$ in $T$ but does not depend on the internal structure of $T$. In the general case our results will rely on a physically plausible, but still unproven, mean field assumption. We shall provide numerical evidence for this assumption.

\subsection{Results and organisation of the paper.}
\label{ssResults}
In section \ref{S:genfunc} we first give the precise definition of our vertex splitting model.
We then study the special case where the splitting probability weights are linear with the initial vertex degree $i$ and focus on the vertex degree distribution. 

In section \ref{s:genfunc} we write exact recurrence equations for the general local vertex degree probability distributions. 
Using the Perron-Frobenius theorem \cite{seneta} we show that the single vertex degree probability distribution $\boldsymbol{\rho}=\{\rho_k\}$ ($\rho_k$ is the density of vertices with a given degree $k$)  has a well
defined limit as the size of the tree goes to infinity.    
We furthermore show that  $\boldsymbol{\rho}=\{\rho_k\}$ is given by an eigensystem equation of the form $\mathbf{B}\boldsymbol{\rho}=\lambda\boldsymbol{\rho}$, where  $\mathbf{B}$ is a matrix depending on the weights of the model. The proof depends on the matrix $\mathbf{B}$ being diagonalizable. Similar techniques have been used to find the asymptotic degree distribution in random recursive trees \cite{Svante}.

In section \ref{s:explicitsolutions} we find explicit solutions to the above eigensystem equation in special cases when the matrix $\mathbf{B}$ is diagonalizable. We study in almost full generality the cases when $d_{\mathrm{max}}=3$  and $d_{\mathrm{max}}=4$ and a particular choice of weights for arbitrary $d_{\mathrm{max}}$. In section \ref{ssDiscWeight} we show that the condition of diagonalizability of $\mathbf{B}$ is not very restrictive.

In section \ref{ssMeanField} we relax the condition of linearity on the splitting weights $w_i$.  We argue that mean field theory is still valid and that the degree probability distribution $\boldsymbol{\rho}$ is still given by the same linear eigensystem equation as in the linear case. 
We give good numerical evidence of the validity of these mean field equations for $d_{\mathrm{max}}=3$ trees.

For infinite $d_{\mathrm{max}}$ and
linear and uniform splitting probabilities we can still calculate the vertex degree
distribution in closed form using mean field theory. This is done in section \ref{ssdinfinite}, where we show that 
it agrees with numerical simulations. The vertex degree distribution 
is found to fall off factorially in this case.  The vertex
degree distribution always becomes independent of the initial tree as time
goes to infinity. 

In section \ref{s:subtree} we study probabilities associated to the local subtree structure of the tree, as seen from any vertex, and as a function of its creation time $s$. More precisely, we are able to write recursion relations for the probability $p_k(\ell_1,\cdots, \ell_k;s)$ that the vertex created at time $s$ is of degree $k$, with the $k$ subtrees with fixed respective sizes $\ell_1,\cdots, \ell_k$.

In section \ref{ssSubTree} we write the detailed (and quite complicated) recursion relations in the linear splitting weights case, and we show that they simplify for the probabilities $p_k(\ell_1,\cdots, \ell_k)$ obtained by summing over the time of creation. 

In section \ref{ss2PtFunc} we derive a recursion relation for a simple two-point function $q_{ki}(\ell_1,\ell_2)$ related to the decomposition of a tree into two subtrees, that will be useful. The substructure probabilities should find several applications in the studies of these trees. 

In section \ref{s:Hausdorff} we study the scaling properties of the trees by computing their Hausdorff dimension (or intrinsic fractal dimension).

In section \ref{ssDefHausdff} we define the radius and the fractal dimension $d_H$ of a tree (since there are several almost equivalent definitions).

In section \ref{ssGeodDist} we explain how the radius (hence the fractal dimension) is obtained from the two-point function defined in \ref{ss2PtFunc}.

In section \ref{ScalingDF}, using a natural scaling hypothesis for the two-point function, we show that the fractal dimension $d_H$ is also  given by the solution of an eigensystem equation of the form 
$\mathbf{C}\boldsymbol{\omega}=w_2/d_H\,\boldsymbol{\omega}$, where $\mathbf{C}$ is a complicated matrix which is a function of weights of the model. We use a Perron-Frobenius argument to prove that this eigensystem equation has a unique physical solution.

In section \ref{ssDFbound} we establish bounds on the fractal dimension. We show that it can vary continuously with the splitting weights between $1$ and $+\infty$.

In section \ref{ss=expd=3} we study the case of  ($d_{\mathrm{max}}=3$) trees with linear uniform splitting weights. We compare the analytical result for $d_H$ with numerical simulations and show that the agreement is good.

In section \ref{ss:HausGen} we relax the condition of linearity on the splitting weights $w_i$.  We argue that mean field theory is still valid and that the fractal dimension $d_H$ is still given by the same linear eigensystem equation as in the linear case. 
We present preliminary numerical evidence of the validity of these mean field equations for binary trees.

In section \ref{s:correlations} we study the correlations between the degrees of neighbouring
vertices.  This amounts to studying the density $\rho_{ij}$ 
of links with vertices of degrees $i$ and $j$.  

In section \ref{ssVVlinear} we write general equations for these correlations in the linear splitting weight case.  Some technical details are left  to appendix \ref{A:corr}. 

In the simple case of $d_{\mathrm{max}}=3$ trees these correlations can be calculated explicitly, and compared with numerical simulations. This is done in section \ref{ssVVd=3}, 
and one sees that there are always nontrivial correlations, i.e.,
$\rho_{ij}\neq {\rho_i\rho_j/ ( 1-\rho_1)^{-1}}$ 
where $\rho_i$ is the density of vertices
of degree $i$ (the r.h.s. obtained if the degrees of neighbouring vertices are completely independent).  Correlations between vertices which are
further away from each other can be studied by the same method.
We show that our results for the degree-degree correlations are different from those for a different class of
statistical random graphs \cite{bialas}.

In section \ref{VVnonlinear} we extend our results for the case of non-linear splitting weights, assuming mean field theory. We show that there is a very good agreement between our analytical results and numerical simulations (still in the case of $d_{\mathrm{max}}=3$ trees).

Finally section \ref{s:Disc} is the conclusion. We discuss there in more detail the relationship between our model and 
other models of random trees, in particular models of phylogenetic trees \cite{alphamodel,AldousBeta,HassMiermontPitmanWinkel,PitmanWinkel,ChenFordWinkel,AldousClad}, and we point out some open questions.
 
\section{The vertex splitting model} 
\label{S:genfunc}
\subsection{Description of the model}
\label{ssDescr}
A rooted planar tree is a tree graph embedded into the plane $\R^2$ which 
contains a distinguished vertex that we call the {\em{root}}.  We will
always assume that the root has order one.  The
planarity condition is for convenience only and can also be
implemented by cyclically ordering the links incident on each vertex.

Let $\mathcal{T}$ be the collection of all rooted planar trees for 
which every vertex has finite degree at most 
\begin{equation}
\label{ddmax}
d_{\mathrm{max}}=d\qquad\text{maximal vertex degree}.
\end{equation}
We shall discuss later the case $d_{\mathrm{max}}=\infty$.
Denote the number of vertices in a tree $T \in \mathcal{T}$ by $|T|$ and the number of vertices of degree $i$ in $T$ by $n_i(T)$. Let $\mathcal{T}_N$ 
be those trees $T \in \mathcal{T}$ with $|T| = N$. Let
$$\mathbf{M} = 	\begin{bmatrix}
	0       & w_{1,2} & w_{1,3} & \cdots& w_{1,d-1} & w_{1,d} \\
	w_{2,1} & w_{2,2} & w_{2,3} & \cdots& w_{2,d-1} & w_{2,d} \\
        w_{3,1} & w_{3,2} & w_{3,3} & \cdots& w_{3,d-1} & 0 \\
	w_{4,1} & w_{4,2} & w_{4,3} & &    0      & 0 \\
 	\vdots  &  \vdots & \vdots  & \udots& \vdots   & \vdots \\
        w_{d,1} & w_{d,2} & 0       & \cdots& 0   & 0
	\end{bmatrix}
	  $$
be a symmetric matrix with nonnegative entries that we call 
\textit{partitioning weights}. 
We define a collection of 
non-negative numbers called \textit{splitting weights},
$w_1,w_2,\ldots , w_d $, by
\begin{eqnarray} \label{100}
w_i &=& \frac{i}{2}\sum_{j=1}^{i+1} w_{j,i+2-j}.
\end{eqnarray}

We now define a growth rule for planar trees which we call 
{\it vertex-splitting}. Given a tree $T\in\mathcal{T}_N$
\begin{enumerate}
\item[(i)] 
Choose a vertex $v$ of $T$ with probability 
${w_i}/{\mathcal{W}(T)}$ where $i$ is the order of $v$ and  
\begin{eqnarray} \label{eqn:the_weight}
	\mathcal{W}(T) &=& \sum_{j=1}^{d}w_j n_j(T).
\end{eqnarray}
\item[(ii)] 
Partition the edges incident with $v$ into two disjoint sets $V$ and 
$V'$ of adjacent edges with probability 
\begin{equation}
\label{eq:pww}
p_{k,i+2-k}=\dfrac{w_{k,i+2-k}}{w_i}\ .
\end{equation}
The set 
$V$ contains $k-1$ of the edges and $V'$ contains $i-(k-1)$ of these 
edges, $k=1,\ldots ,i$. 
For a given $k$, all such partitionings are taken to be 
equally likely.
\item[(iii)] 
Move all edges in $V'$ from $v$ to a new vertex $v'$ and create an edge 
joining $v$ to $v'$.  If $v$ is the root, then the new vertex of order
one is taken to be the root.
\end{enumerate}
This vertex-splitting operation is illustrated in figure \ref {F_split} (the root vertex is circled).
\\[1em]
\begin{figure} [!h]
\centerline{\scalebox{0.8}{\includegraphics{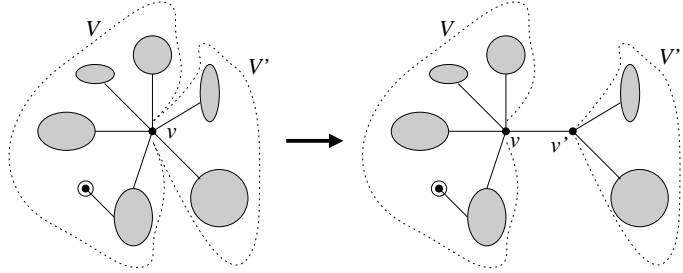}}}
\caption{Illustration of the splitting process for $i=6$ and $k=5$.} \label{F_split}
\end{figure}
After the splitting operation, the degree of vertex $v$ is $k$ and 
the degree of vertex $v'$ is $i+2-k$. 
Since the maximum allowed vertex degree is $d$ we define 
$w_{d+1,1}=w_{1,d+1}=0$, 
i.e.~ we do not allow splittings of vertices of 
degree $d$ that produce vertices of degree $d+1$.
If the partitioning weights are chosen such that 
$w_{i,j} = 0$ for $i\neq 1$ 
or $j\neq 1$, then the vertex-splitting model is equivalent to the 
preferential attachment model discussed in \cite{MASS_DIST}.

In Sections 3 and 4 we will find it convenient to label the vertices
according to their time of creation.  In this case we append the
following to our rules:
\begin{enumerate}
\item[(iv)]
Let $a$ be the label of the vertex chosen in (i).
If $v$ is further away from the root
than $v'$ in step (iii) then we let $v$ keep the label $a$
and give $v'$ the label $N$.
Otherwise label $v$ with $N$ and label $v'$ with $a$.
\end{enumerate}
This book-keeping device has no effect on the dynamics of the model.
The single root vertex (which is the only tree in $\mathcal{T}_1$) has
label 0.  We will often think of the number of vertices (or
equivalently, the number of links) as time and denote it by $t+1$ (or
$t$ in the case of links) assuming we start with the single vertex
tree at time $t=0$.

If the partitioning weights are chosen such that the splitting 
weights are linear, 
\begin{equation} \label{linw}
	w_i = ai + b
\end{equation}
for some $a,b\geq 0$ then the model is easier to analyze since the 
weight of a tree $T \in \mathcal{T}$ depends only on the size 
of the tree
\begin {equation} \label{101}
 	\mathcal{W}(T)  = (2a+b)|T|-2a.  
\end {equation}
This is easily seen from the two constraints on the vertex degrees,
\begin {equation} \label{con}
	\sum_{i=1}^{d}n_i(T) \;=\; |T| \quad \mbox{ and } \quad
	\sum_{i=1}^{d}i n_i(T) \;=\; 2(|T|-1).
\end{equation} 
By abuse of notation, in this case we will write  $\mathcal{W}(|T|) = \mathcal{W}(T)$. We will also sometimes restrict to uniform partitioning weights, i.e.~
\begin{eqnarray}
	w_{i,k+2-i} &=& \left\{
 	\begin{array}{cll}
 	w_k/ {k+1\choose 2} 
	& \text{for $i=1,\ldots,k+1$,} & \mbox{ if } k < d,\\[1em]
 	w_k/ {k\choose 2} 
	& \text{for $i = 2,\ldots,k$,} & \mbox{ if } k = d.
	\end{array} 
	\right.
\end{eqnarray}

\subsection{Distribution of vertex degrees for linear splitting weights} \label{s:genfunc}
Start from a finite tree $T_0$ at time $t_0 =|T_0|$ and perform vertex
splitting according to the rules described in the previous subsection
$\tau$ times. We then obtain a tree in $\mathcal{T}_{t_0+\tau}$. Let $t = t_0 + \tau$.   
The vertex splitting operation induces a probability measure $\mu_t$ on 
$\mathcal{T}_{t}$, which of course depends on the initial
tree $T_0$. In this section we will drop $T_0$ from function arguments with the understanding that it is implied, unless otherwise stated.

Let $p_t(n_1,\ldots,n_d)$ be the probability that 
$T\in\mathcal{T}_{t}$ has $(n_1(T),\ldots, n_d(T)) = (n_1,\ldots , n_d)$ 
according to the measure $\mu_t$. We wish to find the mean value of $n_k(T)$ for $T\in {\mathcal T}_t$ 
with respect to the measure $\mu_t$. Denote this value by $\overline{n}_{t,k}$. We define the vertex degree densities $\rho_{t,k} \equiv  \overline{n}_{t,k}/t$ and with some conditions on the partitioning weights we will prove the existence of the limit $$\lim_{t \rightarrow \infty}\rho_{t,k} \equiv \rho_k$$ and show that the $\rho_k$ satisfy a system of linear equations.

Let ${\bf x}=(x_1,\ldots ,x_d)\in {\mathbb R}^d$ and define the probability generating function
\begin{eqnarray} 
	\mathcal{H}_t({\bf x}) &=& 
	\sum_{n_1+\cdots n_d=t} 
		p_t(n_1,\ldots,n_d ) x_1^{n_1}\cdots x_d^{n_d} \label{103}
\end{eqnarray}
\begin {proposition}
The probability genereting function $\mathcal{H}_t({\bf x})$ satisfies the recurrence 
\begin{eqnarray} \label {recur}
	\mathcal{H}_{t+1}(\bx ) &=& 
	\sum_{n_1 + \cdots +n_d=t} 
	\dfrac{p_t(n_1,\ldots,n_d)}{\sum_{i=1}^{d} 
	n_i w_{i}} \, \bc (\bx ) \cdot 
	\nabla ( x_1^{n_1} \cdots x_d^{n_d})
\end{eqnarray}
for all $t \geq t_0$, where
\begin{eqnarray}
	 \bc (\bx ) &=& 
	\left(c_1(\bx ),c_2(\bx ),\ldots,c_{d}(\bx )\right)
\end{eqnarray}
with
\begin{eqnarray} \label{csum}
	c_i(\bx ) &=& 
	\dfrac{i}{2} \sum_{j=1}^{i+1} w_{j,i+2-j} x_j x_{i+2-j}
\end{eqnarray}
and $\nabla = \Big(\partial/\partial x_1,\ldots,
 \partial/\partial x_d\Big)$ is the standard gradient operator.
\end {proposition}

\begin {proof}
Any tree contributing to $\mathcal{H}_{t+1}$ can be obtained by splitting a vertex in a tree on t vertices. This process can be divided into three steps:
\begin{enumerate}
 \item[(i)] Choose a tree $T \in \mathcal{T}_t$ with vertex degree distribution $(n_1,\ldots,n_d)$ with probability $p_t(n_1,\ldots,n_d)$.
 \item[(ii)] Select a vertex in $T$ of degree $i$ with probability $n_iw_i/\sum_j{n_jw_j}.$
 \item[(iii)]  Partition the edges incident to the chosen vertex into two sets $V$ and $V'$ of adjacent edges with $j-1$ and $i+1-j$ elements, respectively, with probability $iw_{j,i+2-j}/w_i$ if $j \neq i+2-j$ and with probability $\frac{i}{2} w_{j,i+2-j}/w_i$ if $j = i+2-j$. In the latter case there is a symmetry between $V$ and $V'$ which accounts for the factor $1/2$.
\end{enumerate}
Multiplying together the probabilities in (i)--(iii) gives the probability of removing a vertex of degree $i$ and creating two new vertices of degree $j$ and $i+2-j$. In terms of the generating function this amounts to replacing $x_1^{n_1}\cdots x_d^{n_d}$ by $x_{i}^{-1} x_{j} x_{i+2-j} x_1^{n_1}\cdots x_d^{n_d}$. The probability is 

\begin {equation} \label{probgen}
\frac{p_t(n_1,\ldots,n_d)}{\sum_j{n_jw_j}} n_i \times \left\{
 	\begin{array}{cll}
 	i w_{j,i+2-j} 
	& \text{$\mathrm{if}$ $j \neq i+2-j$,} & \\[1em]
 	\frac{i}{2} w_{j,i+2-j} 
	& \text{$\mathrm{otherwise.}$} &
	\end{array} 
	\right.
\end{equation} 

The partial derivative $\partial / \partial x_i$ in $\nabla$ takes care of removing a vertex of degree $i$ and provides the factor $n_i$. In $c_i(\mathbf{x})$, the factors $x_j x_{i+2-j}$ add two vertices of degree $j$ and $i+2-j$ respectively and the appropriate weights are given. Now sum (\ref{probgen}) over $j= 1,\ldots,i+1$,  i.e.~over all possible partitionings in (iii), to obtain the sum in (\ref{csum}). Note that terms for which $j \neq i+2-j$ appear twice in the sum, since $w_{p,q}$ is symmetric in $p$ and $q$, and the term for which $j = i+2-j$ appears once. This explains the factor $1/2$ in front of the sum in (\ref{csum}). The dot product of $\mathbf{c}(\mathbf{x})$ and $\nabla$ accounts for the sum over all vertex degrees, and finally sum over all vertex degree configurations in the initial tree to obtain (\ref{recur}.\end {proof}

For linear weights (\ref{linw}), equation (\ref{recur}) reduces to
\begin{equation} \label{simprec}
	\mathcal{H}_{t+1}(\bx ) = 
	\frac{1}{\mathcal{W}(t)} 
	\bc (\bx )\cdot\nabla\mathcal{H}_{t}(\bx ).
\end{equation}
The remainder of this subsection concerns linear weights only.
 We have
\begin{equation}
	\overline{n}_{t,k} \;=\;
	\sum_{n_1+...+n_d = t}p_t(n_1,...,n_d)n_k \; =\;
	 \partial_k{\mathcal H}_t(\bx )|_{\bx = {\mathbf 1} },
\end{equation}
where ${\mathbf 1} =(1,1,\ldots ,1)$.
To get a recursion equation for $\overline{n}_{t,k}$, differentiate 
both sides of (\ref{simprec}) with respect to $x_k$ and set 
$\bx = {\mathbf 1} $ to find
\begin{eqnarray} \nonumber
	\overline{n}_{t+1,k} &=& 
	\frac{1}{\mathcal{W}(t)}
	\left(\sum_{i=k-1}^{d} i w_{k,i+2-k} \overline{n}_{t,i}  
	 + \sum_{i=1}^d w_i \partial_i \partial_k 
	 \mathcal{H}_t(\bx )|_{\bx = {\mathbf 1} }\right).\\ \label{cr}
\end{eqnarray} 
Since the weights are linear we can use the constraints in (\ref{con}) 
to rewrite the last term in \rf{cr} as
\begin{eqnarray}
	\sum_{i=1}^d w_i \partial_i \partial_k \mathcal{H}_t
	(\bx )|_{\bx = {\mathbf 1}} &=& 
	\left(-w_k + \mathcal{W}(t)\right)\overline{n}_{t,k}.
 \end{eqnarray} 
Inserting this into (\ref{cr}) we see that the equations close 
\begin{eqnarray}  \label{aresult}
	\overline{n}_{t+1,k} &=& 
	\dfrac{1}{\mathcal{W}(t)}\left( -w_k\overline{n}_{t,k} + 
	 \sum_{i=k-1}^{d} i w_{k,i+2-k} \overline{n}_{t,i} \right) 
	 +\overline{n}_{t,k}.  
\end {eqnarray}
 We can also write the recursion
 in terms $\rho_{t,k}$ and find
\begin {equation} \label{result}
	\rho_{t+1,k} = 
	\frac{t}{\mathcal{W}(t)}\left(-w_k\rho_{t,k}+\sum_{i=k-1}^{d}i 
	 w_{k,i+2-k} \rho_{t,i} \right) + t(\rho_{t,k}-\rho_{t+1,k}).
\end {equation}
The above equation can be put in the matrix form
\begin{eqnarray}
	\mbox{\boldmath $\rho$}_{t+1} &=& \mathbf{A}_t \mbox{
	\boldmath $\rho$}_t
\end{eqnarray}
where
\begin{eqnarray} \nonumber
	\mbox{{\boldmath $\rho$}}_t &=& 
	\left(\rho_{t,1},\rho_{t,2},\ldots,\rho_{t,d}\right)^T, 
	\hspace*{4ex} 
	\mathbf{A}_t \;=\; \frac{t}{t+1}\Bigg(\mathbf{I} + 
	\frac{1}{\mathcal{W}(t)}\mathbf{B}\Bigg), \\
\end{eqnarray}

\begin{eqnarray} \nonumber
	\mathbf{B} = {\footnotesize{\left[\begin {array} {ccccll}    
  w_{1,2}&2w_{1,3}&\cdots&(d-2)w_{1,d-1}&(d-1)w_{1,d}& \multicolumn{1}{c}{0} \\
  w_{2,1}&2w_{2,2}&\cdots&(d-2)w_{2,d-2}&(d-1)w_{2,d-1}&dw_{2,d}\\
  0&2w_{3,1}&\cdots&(d-2)w_{3,d-3}&(d-1)w_{3,d-2}&d w_{3,d-1}\\
  \vdots & \ddots & \ddots & \vdots & \multicolumn{1}{c}{\vdots} & \multicolumn{1}{c}{\vdots}  \\
  \vdots &  & \ddots & (d-2)w_{d-1,1} & (d-1)w_{d-1,2} & d w_{d-1,3} \\
  0 & \cdots & 0 & 0 & (d-1)w_{d,1} & d w_{d,2}
	\end {array}\right]}} 
	- \text{diag}(w_i)_{1\leq i\leq d} \\ \label{thematrix}
\end{eqnarray}
and $\mathbf{I}$ is the identity matrix.

If we denote the vertex degree densities of the initial tree $T_0$
by $\mbox{\boldmath $\rho$}_{t_0}$ we can write the densities for trees on 
$t$ vertices which grow from the initial tree as
\begin{equation} \label{finite_n_solution}
	\mbox{\boldmath $\rho$}_{t} = 
	\left(\prod_{i=t_0}^{t-1} \mathbf{A}_i\right) \mbox{\boldmath $\rho$}_{t_0} 
	 \;\;=\;\; \frac{t_0}{t} \left(\prod_{i=t_0}^{t-1} 
	 \left(\mathbf{I}  + \frac{1}{\mathcal{W}(i)}\mathbf{B}\right)
	 \right) \mbox{\boldmath $\rho$}_{t_0}.
\end{equation}
We will establish convergence of the right hand side by imposing some technical restrictions on $\mathbf{B}$. It turns out that the limiting distribution is independent of the initial distribution $\mbox{\boldmath $\rho$}_{t_0}$. We begin with some necessary lemmas.

\begin {lemma} \label{sumrules}
If $\lambda$ is an eigenvalue of $\mathbf{B}$ with corresponding eigenvector
\\ ${\mathbf e}_\lambda =(e_{\lambda 1},\ldots ,e_{\lambda d})$, i.e.
\beq{eigenvalueeq}
{\mathbf B}{\mathbf e}_\lambda = \lambda {\mathbf e}_\lambda,
\eeq
then the following holds:
\begin {eqnarray}
\lambda \sum_{i=1}^d  e_{\lambda i} &=& \sum_{i=1}^d w_i
e_{\lambda i} \quad\quad \text{and}\label{2.23}\\
\lambda \sum_{i=1}^d i e_{\lambda i} &=& 2 \sum_{i=1}^d w_i
e_{\lambda i}.\label{2.24}
\end {eqnarray}

\begin {proof}
We prove the second identity.  The first
identity is established by a
similar calculation.  
Multiply the $i$-th component of the eigenvalue
equation \rf{eigenvalueeq} by $i$ and sum over $i$ to get
\begin {eqnarray} 
\lambda \sum_{i=1}^d i  e_{\lambda i} &=& -\sum_{i=1}^d i w_i
e_{\lambda i} + \sum_{i=1}^d i \sum_{k=i-1}^d k w_{i,k+2-i}
e_{\lambda k} \nonumber \\
&=& -\sum_{i=1}^d i w_i  e_{\lambda i} + \sum_{k=1}^d k
\left(\sum_{i=1}^{k+1} i w_{i,k+2-i} \right) e_{\lambda k} .
\end {eqnarray}
Using $w_{i,j} = w_{j,i}$ we find that 
\begin{equation}
\sum_{i=1}^{k+1} i
w_{i,k+2-i} = \frac{k+2}{2} \sum_{i=1}^{k+1} w_{i,k+2-i}
\end{equation} 
and this together with the
definition of the splitting weights \rf{100} 
proves the identity.
\end {proof}

\end {lemma}

\begin{lemma} \label {perronfrob}
If

\begin{enumerate}
 \item  $w_{k,1}=w_{1,k} > 0$ for $k = 1,\ldots,d$ (i.e.~ it is possible to
produce vertices of maximal degree $d$ from vertices of degree $j<d$) and
\item  $w_{i,d+2-i} >0$ for at least one $i$ with $2\leq i \leq d-1$  (i.e.~it is possible to split vertices of maximal degree $d$),
\end{enumerate}
then $w_2$ is a positive, simple eigenvalue of $\mathbf{B}$.  All
other eigenvalues of $\mathbf{B}$ have a smaller 
real part. The corresponding eigenvector ${\mathbf e}_{w_2}$ can be
taken to have all entries positive. 
\end{lemma}

\begin {proof}
We begin by choosing a number $\gamma > \max_{1\leq k \leq d} 
\left\{  w_k-kw_{k,2}\right\}$ and define $\mathbf{P} =
\mathbf{B}+\gamma \mathbf{I}$.  The matrix 
$\mathbf{P}$ has only non-negative entries and
the conditions (1) and (2) on $\mathbf{B}$ guarantee that it is
primitive, i.e. there is a number $k$ such that all entries of the
matrix $\mathbf{P}^k$ are positive. Therefore, by the
Perron-Frobenius theorem \cite{seneta},  
$\mathbf{P}$ has a simple positive eigenvalue $r$ and all other
eigenvalues of $\mathbf{P}$ have a smaller modulus. The corresponding 
eigenvector ${\mathbf e}_r$ can be taken to have all entries
positive. We normalize the eigenvector such that 
\beq{109}
\sum_{i=1}^d e_{ri} = 1.
\eeq

Shifting back to the matrix $\mathbf{B}$ we find that $w \equiv r-\gamma$ is a
simple real eigenvalue of $\mathbf{B}$ with the largest real part and the
corresponding eigenvector is ${\mathbf e}_w = {\mathbf e}_r$. We see 
right away from \rf{2.23} and with the chosen normalization that 
\begin {equation} \label{eq_w}
w = \sum_{i=1}^d w_i  e_{wi}.
\end {equation}
Since the weights are linear, Lemma \ref{sumrules} shows that $w = w_2$.
\end {proof}

Note that the first condition on the weights in the above lemma is natural since
we have fixed a maximal degree $d$ and therefore we want to be able to produce 
vertices of
degree $d$. The second condition, however, does not seem to be necessary for the
results to hold but we still require it in order to use the Perron-Frobenius
theorem for primitive matrices. This condition is not very restrictive
in the case of linear weights since it holds for all $a$ and $b$
except when $ad+b = 0$.

\begin{lemma} \label{limit} Let $\lambda\in {\mathbb C}$. Then

 \begin{equation}
\frac{t_0}{t} \prod_{i=t_0}^{t-1} 
\left(1  + \frac{1}{\mathcal{W}(i)}\lambda\right) \longrightarrow \left\{
 	\begin{array}{cll}
 	\frac{t_0w_2}{t_0w_2-2a} 
	& \text{$\mathrm{if}$ $\lambda = w_2$,} & \\[1em]
 	0 
	& \text{$\mathrm{if}$ $\mathrm{Re}(\lambda) < w_2$} &
	\end{array} 
	\right.
\end{equation}
\\
 as $t\to\infty$.
\end{lemma}

\begin {proof}
The result follows from the identity

\begin {eqnarray}
\frac{t_0}{t} \prod_{i=t_0}^{t-1} \left(1  +
\frac{1}{\mathcal{W}(i)}\lambda\right) &=& \frac{t_0}{t} \frac{\Gamma  \left( t-{\frac {2\,a-\lambda}{w_{{2}}}} \right) 
\Gamma  \left( t_{{0}}-\,{\frac {2a}{w_{{2}}}} \right)}{
  \Gamma  \left( t-\,{\frac {2a}{w_{{2}}}} \right)  \Gamma  \left( t_{{0}}-{\frac {2\,a-\lambda}{w_{{2}}}}
 \right) }.
\label{ppp}
\end {eqnarray}

\end {proof}

\begin {theorem}   
With the assumptions on $\mathbf{B}$ in Lemma \ref{perronfrob} and
the additional 
assumption that $\mathbf{B}$ is diagonalizable, the limit as
$t\to\infty$ of the right hand side
of equation (\ref{finite_n_solution}) exists and is given by the eigenvector
${\mathbf e}_{w_2}$ of $\mathbf{B}$ normalized such that
\begin {equation} \label{norm}
\sum_{i=1}^d  e_{w_2i} = 1.
\end {equation}
\end {theorem}
\begin {proof}
We use the normalization in (\ref{norm}) and expand $\mbox{\boldmath $\rho$}_{t_0}$ in the basis of 
eigenvectors of $\mathbf{B}$. Using the results of Lemmas \ref{sumrules} and \ref{perronfrob} and that $T_0$ satisfies the equations in (\ref{con}) we see that the expansion is of the form

\begin {equation}
\mbox{\boldmath $\rho$}_{t_0} = \frac{w_2 t_0 - 2a}{w_2 t_0} \mathbf{e}_{w_2} + \sum_{i = 1}^{d-1} a_i \mathbf{e}_{\lambda_i}
\end {equation}
\\
where $\lambda_i, ~i=1,\ldots,d-1$ are the eigenvalues of $\mathbf{B}$ with real part less than $w_2$.  The result now follows from Lemma \ref{limit}.
\end {proof}

Theorem 2.4 shows that with the above conditions on $\mathbf{B}$ the limit of the
vertex degree densities exists, is independent of the initial tree
and is given by
\beq{qqq}
\mbox{\boldmath $\rho$} \equiv
\lim_{n\rightarrow \infty} \mbox{\boldmath $\rho$}_t = {\mathbf e}_{w_2}.
\eeq
The limiting densities 
are therefore the unique positive solution to equation (\ref{eigenvalueeq}), i.e. ~
\begin{equation} \label{exactsolution}
	\rho_k = 
	-\frac{w_k}{w_2} \rho_k 	
	+ \sum_{i=k-1}^{d}i \frac{w_{k,i+2-k}}{w_2} \rho_{i}.
\end{equation}
\\
\subsection {Explicit solutions} 
\label{s:explicitsolutions}
 We discuss three simple special cases.

1) When $d=3$ we find that 

\begin{equation}
	\mathbf{B} = {\begin {bmatrix}    
  0&2w_{1,3}&0\\[0.3em]
  w_{2,1}&w_{2,2}-2w_{3,1}&3w_{3,2}\\[0.3em]
  0&2w_{3,1}&0\\
	\end {bmatrix}}.  \label{B3}
\end{equation}
\\
If the weights satisfy the conditions in Lemma \ref{perronfrob} it is easy to see that $\mathbf{B}$ is diagonalizable. For linear
splitting weights $w_i = ai+b$ and uniform partitioning weights the positive solution 
of \rf{exactsolution} is
\beq{val}
\rho_1 = \rho_3 = \frac{2}{7} \quad {\rm and} \quad \rho_2 = \frac{3}{7}
\eeq 
for all
values of $a$ and $b$ as can easily bee seen from the simple structure
of the $\mathbf{B}$ in this case.

2) When $d = 4$, the splitting weights linear and the partitioning weights
uniform one can check that 

\begin{equation}
	\mathbf{B} = {\begin {bmatrix}    
  0&\frac{2}{3}(2a+b)&\frac{1}{2}(3a+b)&0\\[0.3em]
  a+b&-\frac{1}{3}(2a+b)&\frac{1}{2}(3a+b)&\frac{2}{3}(4a+b)\\[0.3em]
  0&\frac{2}{3}(2a+b)&-\frac{1}{2}(3a+b)&\frac{2}{3}(4a+b)\\[0.3em]
0&0&\frac{1}{2}(3a+b)&-\frac{1}{3}(4a+b)\\
  
	\end {bmatrix}}. \label{B4}
\end{equation}
\\
When $4a+b > 0$ the weights satisfy the conditions in Lemma \ref{perronfrob}. The eigenvalues of $\mathbf{B}$ are
$-\frac{1}{12}(33a+13b\pm\sqrt{a^2-78ab-15b^2})$, $w_2$ and 0. This shows that
$\mathbf{B}$ is diagonalizable except when $a/b = 39\pm16\sqrt{6}.$ One can
analyze these cases separately using a basis of generalized eigenvectors and show
that the right hand side of equation (\ref{finite_n_solution}) still converges
to ${\mathbf e}_{w_2}$.

3) Fix a maximal degree $d$. Choose partitioning weights 
\begin {align*}
w_{1,i} &= w_{i,1} = (i-1)^{-1}, \quad i=2,\ldots,d,  \\
w_{2,d} &= w_{d,2} = d^{-1}  
\end {align*}
and all other weights equal to zero. The splitting weights are then $w_i = 1$ for $i=1,\ldots,d$. Note that if we take the limit $d \rightarrow \infty$ we get a special case of the preferential attachment model.  These weights satisfy the conditions in Lemma \ref{perronfrob}. The nonzero matrix elements of $\mathbf{B}$ are
\begin {equation}
B_{i+1,i} = B_{1,i} =  -B_{i,i} = B_{2,1} = B_{2,d} = 1, \quad 1 < i < d. \\
\end {equation}
The characteristic polynomial of  $\mathbf{B}$ is 

\begin {equation}
p_d(\lambda) = (-1)^{d}\left(1-\lambda\right)\left(1-\left(1+\lambda\right)^{d-1}\right)
\end {equation}
\\
which can easily be proved by induction. The roots of the characteristic polynomial are $\lambda = 1$ and $\lambda = \exp{(\frac{2 \pi i k}{d-1})}-1$, $k=1,\ldots,d-1$ and they are all distinct which shows that $\mathbf{B}$ is diagonalizable. The solution to (\ref{exactsolution}) is
\begin{equation}
\rho_k = {\frac {{2}^{d-k+\delta_{kd}-1}-\delta_{k1}}{{2}^{d-1}-1}}, \quad k=1,\ldots,d.
\end{equation}

\subsection{Discussion of the conditions on the weights} 
\label{ssDiscWeight}
It is not obvious how restrictive the condition that $\mathbf{B}$ must be diagonalizable is regarding the collection of weights one can consider. In the previous subsection we saw that for $d=3$ and $d=4$ the condition was not very restrictive. Also we saw that for every $d$ there is at least one choice of weights which satisfies the conditions in Lemma \ref{perronfrob} and yields a diagonalizable matrix $\mathbf{B}$. We will now show that this guarantees that almost all weights give a diagonalizable $\mathbf{B}$.

Fix a maximal degree $d$. Let $\mathcal{B}_d$ be the set of matrices $\mathbf{B}$ which correspond to partitioning weights that give linear splitting weights and satisfy the conditions in Lemma \ref{perronfrob}. It is clear that if $\mathbf{B},\mathbf{B}'\in \mathcal{B}_d$ then \\ $t\mathbf{B}+(1-t)\mathbf{B}'\in\mathcal{B}_d$ for all $t \in [0,1]$ and so $\mathcal{B}_d$ is convex. Let 
$$\mathcal{B}'_d = \left\lbrace \mathbf{B}\in \mathcal{B}_d ~|~ \mathbf{B} \text{ is diagonalizable.}\right\rbrace.$$ 
From the previous subsection we know that $\mathcal{B}'_d \neq \emptyset$. Since $\mathcal{B}_d$ is convex and $\mathcal{B}'_d \neq \emptyset$ then by Corollary 1 in \cite{diagonal}, $\mathcal{B}'_d$ is dense in $\mathcal{B}_d$. 

We believe that it is possible to extend the result of convergence of the right hand side of (\ref{finite_n_solution}) to all partitioning weights giving  linear splitting weights, relaxing both the condition of diagonalizability of $\mathbf{B}$ and condition (2) in Lemma \ref{perronfrob}. We also believe, in view of simulations, that equation (\ref{exactsolution}) even describes
correctly the vertex degree distribution for
non-linear splitting weights and for the case $d=\infty$. 
In the special case of the preferential attachement model the vertex degree distribution can be calculated by another method \cite{Svante}.
We will look at this more closely in the next two subsections. 

\subsection{Mean field equation for general weights}
\label{ssMeanField}
To generalize equation (\ref{exactsolution}) beyond the case of linear splitting
weights we notice that Lemmas \ref{sumrules} and \ref{perronfrob} do
not rely on the linearity of the weights except in the conclusion of Lemma
\ref{perronfrob} where we show that $w = w_2$.
We therefore conjecture 
that in general the limiting vertex degree densities are the
unique positive solution to

\begin{equation} \label{mfequation}
	\rho_k = 
	-\frac{w_k}{w} \rho_k 	
	+ \sum_{i=k-1}^{d}i \frac{w_{k,i+2-k}}{w} \rho_{i},
\end{equation}

subject to the constraints 
\begin{equation} 
	\label{tme2} \rho_1+\ldots + \rho_d = 1 
\end{equation}
\begin{equation}
	\label{tme3} w_1\rho_1+\ldots + w_d \rho_d = w.
\end{equation}
Recall that $w$ is the unique simple positive eigenvalue of $\mathbf{B}$ defined in (\ref{thematrix}) with the largest real part of all the eigenvalues and $\rho_k$,  $k=1,\ldots,d$ are the components of the associated eigenvector with the proper normalization.

The existence and uniqueness of a positive solution to \rf{mfequation}
satisfying \rf{tme2} and \rf{tme3} follows from the
Perron-Frobenius argument in the proof of Lemma 2.2.  In order to 
distinguish \rf{mfequation} from (\ref{exactsolution}) we refer to it 
as the {\em{mean field equation for vertex degree densities}}. 
One can also arrive directly at this equation by assuming that for 
large $t$ an equilibrium with small enough fluctuations is established, 
and then performing the splitting procedure on this equilibrium.

The solution to the mean field equation for the $d=3$ model and uniform
partitioning weights is
\begin{eqnarray}
	\rho_3 &=& 
	\dfrac{7 \alpha-\sqrt {\alpha \left(\alpha+24\,\beta +24\right) }}{6(2\alpha-\beta-1)}
		\label{mfegw3}
\end{eqnarray}
where $\alpha = \dfrac{w_2}{w_1}$ and $\beta = \dfrac{w_3}{w_1}$. Note that from the constraints we have $\rho_1=\rho_3$ and 
$\rho_2=1-2\rho_3$. This solution (and solutions in general) only 
depends on the ratio of the weights. 
In Figure ~\ref{general:dinfinity} we compare the above
solution to simulations.

\begin{figure}[!h]
\centerline{\scalebox{0.8}{\includegraphics{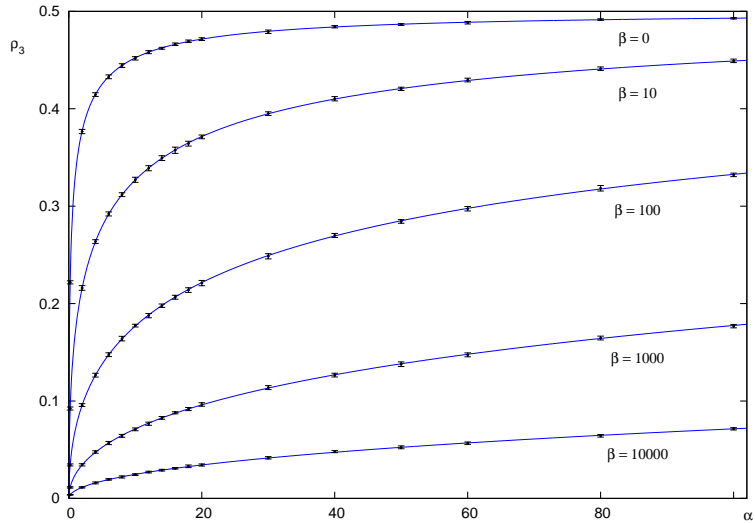}}}
\caption{ The value of $\rho_3$ as given in (\ref{mfegw3})
compared to results from simulations. Each point is calculated from 20 trees on
10000 vertices.} 
\label{general:dinfinity}
\end{figure}

\subsection{The $d_{\mathrm{max}}=\infty$ model with linear weights}
\label{ssdinfinite}

In this subsection we drop the assumption that there is an upper bound
on the vertex degrees but we still assume that all vertex degrees are
finite.
If we assume that equation (\ref{exactsolution}) holds for $d=\infty$,
then it is possible to find an exact solution in the case of linear splitting
weights, $w_i = ai+b$, and uniform partitioning weights. 
Equation (\ref{exactsolution}) becomes
\begin{equation} \label {dinf}
	\rho_k = 
	-\frac{w_k}{w_2} \rho_k + \sum_{i=k-1}^{\infty} \frac{2}{i+1} 
	\frac{w_{i}}{w_2} \rho_{i}.
\end{equation}
Subtracting from this the same equation for $\rho_{k+1}$ we find
\begin{equation}
	\rho_k \left(1+\frac{w_k}{w_2}\right) 
	- \rho_{k+1}\left(1+\frac{w_{k+1}}{w_2}\right) 
	= \frac{2}{k}\frac{w_{k-1}}{w_2}\rho_{k-1}.\label{diff}
\end{equation}
Let $x=b/a$.  The recursion \rf{diff} has the solution  
\begin{eqnarray*}
	\rho_k(x) &=& \left\{
 	\begin{array}{cll}
 	 \displaystyle \frac{2}{C(-1)}
	& \mbox{if } x = -1 \mbox{ and } k = 1\\[1em]
 	\displaystyle \frac{1}{C(x)}\frac{2^{k-1}\Gamma\left(k+x\right)}{\Gamma\left(k\right)
	\Gamma\left(k+3 +2x\right)} \left(k+1+2 x\right)
	&  \mbox{otherwise,}
	\end{array} 
	\right.
\end{eqnarray*}
where
\begin{equation}
C(x) =  \frac{e \sqrt{\pi } ~2^{-\frac{3}{2}-x}  
   I_{\frac{1}{2}+x}(1)}{2 + x}                   
\end{equation}
is a normalization constant such that 
$\sum_i\rho_i = 1$. Here, $I_\nu$ is the modified Bessel function of 
the first kind. The variable $x$ can take values from $-1$ to 
$\infty$. 
The asymptotic behaviour of $\rho_k(x)$ for large $k$ is
\begin {equation}
 \rho_k(x) = \frac{1}{C(x)} \frac{1}{k!}2^{k-1}k^
 {-1-x}\left(1+O\left(\frac{1}{k}\right)\right).\end {equation}
 
The special case $x=\infty$ corresponds to constant 
weights for which the solution is
\begin {equation}
\rho_k(\infty) = \frac{1}{e}\frac{1}{(k-1)!}.\label{exact}
\end {equation}
In Figure \ref{f:dinfinity} we compare the above solutions to simulations 
for five different values of $x$. 

\begin{figure}[!h]
\centerline{\scalebox{0.75}{\includegraphics{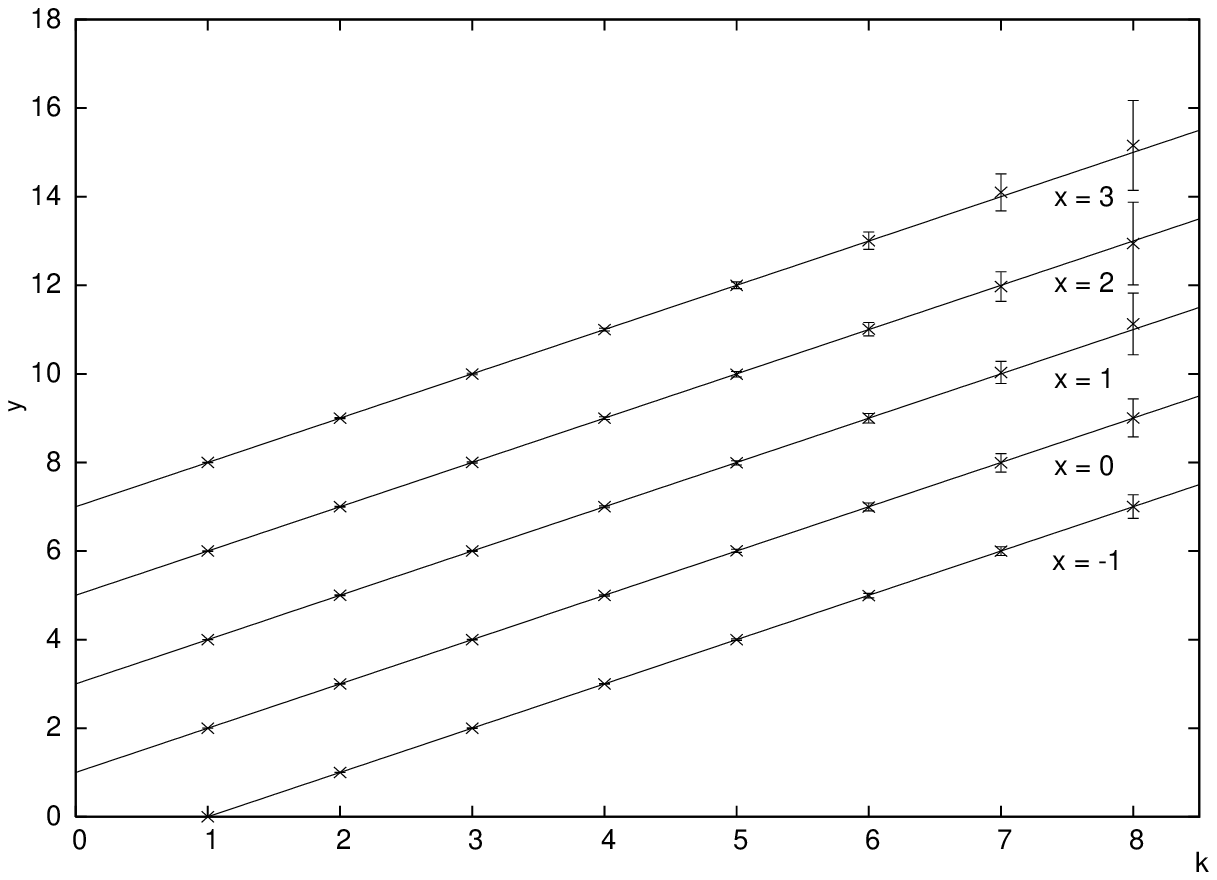}}}
\caption{The solid lines are $y = k+1+2x$ plotted against $k$ for five different 
values of $x$. The datapoints on the graph are calculated from simulations of 100 trees on $10^6$ vertices. For a given $k$ and $x$ they are calculated from the measured density $\rho_{k,\text{meas.}}(x)$ by 
$$
y = C(x)\frac{\Gamma\left(k\right)
	\Gamma\left(k+3 +2x\right)}{2^{k-1}\Gamma\left(k+x\right)}\rho_{k,\text{meas.}}(x)
$$
with an obvious modification if $x=-1$.} \label{f:dinfinity}
\end{figure}

\section{Subtree structure probabilities} 
\label{s:subtree}

In this section we consider the model in which vertices are labeled 
with their time of creation as explained in the definition of the
splitting process (item (iv)). For convenience we will take 
time to be the number of edges in a tree, denoted by $\ell$, starting from the single
vertex tree at time $0$ and adding one link and one vertex at each
time step.  We consider only linear splitting weights $w_i = ai+b$ but comment on
generalizations in the next section. To simplify the notation we define
\begin {equation}
 W(\ell) \equiv \mathcal{W}(T)-w_1 = (2a+b)\ell-a
\end {equation}
where the last equality follows from the linearity of the weights.

We derive exact expressions for probabilities of particular subtree
structures as seen from the vertex created at a given time.
By averaging over these probabilities and assuming the existence of 
a scaling limit,
we shall show how to extract the Hausdorff dimension of the tree and derive bounds on this dimension. In special 
cases we give an exact expression for the Hausdorff dimension.

\subsection {Probabilities related to subtree structure}
\label{ssSubTree}
Consider a tree on $\ell$ edges generated with the splitting procedure
starting from the single vertex tree at time $0$.
Let $p_R(\ell;s)$ be the probability that the vertex created at time $s$ 
is the root.  If $s<\ell$ we find that
\begin{equation}
 p_R(\ell;s) =
         \frac{1}{W(\ell-1)+w_1}W(\ell-1)p_R(\ell-1;s), \label{e_r1} 
\end{equation}
since we can split any vertex except the root in order to get from a
tree at time $\ell -1$ to a tree at time $\ell$. This contributes the factor $p_R(\ell-1;s)$ to $ p_R(\ell;s)$.  Similarly,
\begin{equation}
 p_R(\ell;\ell) =
         \frac{1}{W(\ell-1)+w_1}\sum_{s=0}^{\ell-1}w_1p_R(\ell-1;s),
	                 \label{e_r2}
\end{equation}
since if we create a new root vertex at time $\ell$ the previous root
vertex, labelled $s$ in \rf{e_r2} 
could have been created at any time before $\ell$.  We depict
these processes in Fig.\ \ref{f_r1}.
\begin{figure}[!h]
	\centerline{\scalebox{0.7}{\includegraphics{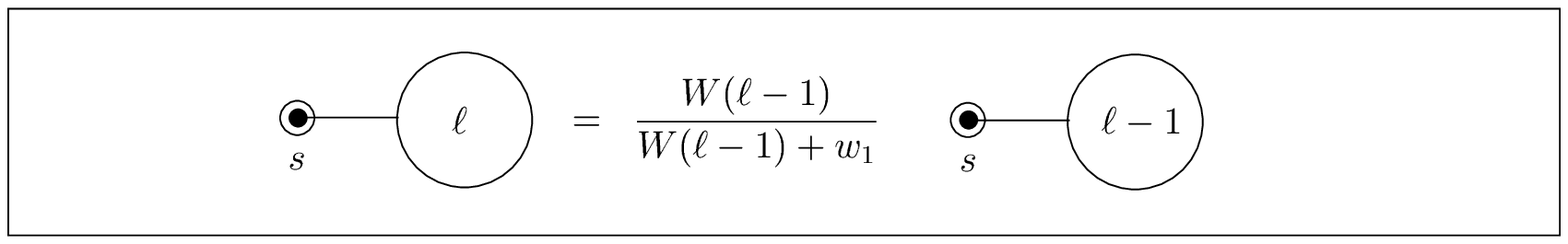}}}
	\centerline{\scalebox{0.7}{\includegraphics{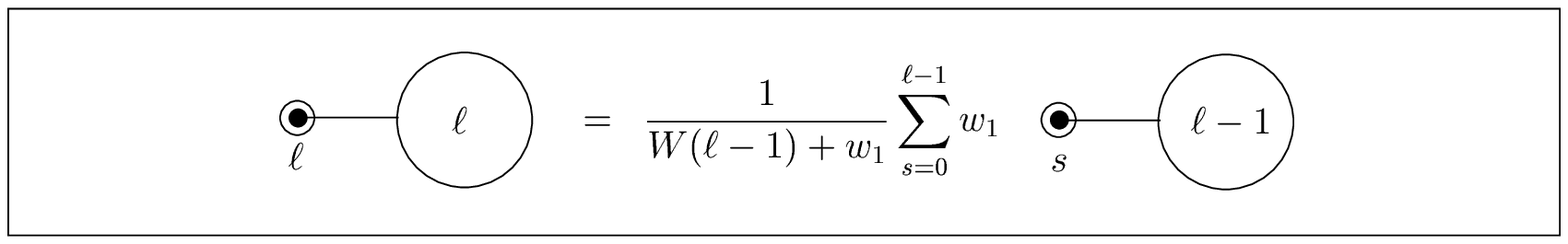}}}
\caption{Diagrams representing equations (\ref{e_r1}) 
and (\ref{e_r2}).} \label{f_r1}
\end{figure}

If $v$ is a vertex of order $k$ in a tree $T$, then there is a unique
link $\ell_1$ incident on $v$ leading towards the root (unless $v$ is the
root).  Let $\ell_2,\ldots ,\ell_k$ be the other links incident on $v$.
The largest subtree of $T$ which contains the root and $\ell_1$ but none
of the links $\ell_i$ with $i\geq 2$ will be called the {\it left subtree} 
(with respect to $v$).
The maximal subtrees which contain one $\ell_j$ with $j\neq 1$ and no
other link $\ell_i$ will be called the {\it right subtrees} (with respect
to $v$).   If $k=1$ then there are of course no right subtrees and if $v$
is the root then we view the left subtree as being empty.  Let $p_k(\ell_1,\ldots,\ell_k;s)$ denote the probability that the vertex 
created at time $s$ has a left subtree on $\ell_1$ edges and right subtrees 
on $\ell_2,\ldots,\ell_k$ edges, 
where $\ell_1 + \ldots + \ell_k = \ell$.
By the nature of the splitting operation, 
$p_k(\ell_1,\ell_2,\ldots,\ell_k;s)$ is symmetric under 
permutations of $(\ell_2,\ldots,\ell_k)$.  We will sometimes refer to
the vertex created at time $s$ as the $s$-vertex.

By the definition of the relabeling when we split we have
\begin{equation}
	p_1(\ell;\ell) = 0,   \label{e_r23}
\end{equation}
because the vertex closer to the root gets a new label and
therefore no leaf except the root 
can have the maximal label.  In the case $s<\ell$ we
find the recursion 
\bea
p_1(\ell;s) & = & \frac{1}{W(\ell-1)+w_1}\Big[W(\ell-1)
	p_1(\ell-1;s)  \nonumber \\ \nonumber
    	&& \hspace*{2ex} \;+\; \sum_{i=1}^{d-1}iw_{i+1,1}
	\sum_{\ell'_1+\ldots+\ell'_i = \ell-1}
	p_i(\ell'_1,\ldots,\ell'_i;s) 
	+ \delta_{\ell1}w_1\Big]. \\ \label{e_r3}
\eea
The first term in the square bracket 
corresponds to the case when we do not split the vertex
with label $s$. The second term corresponds to splitting the
$s$-vertex which can have any order up to $d-1$. Finally the last
term corresponds to the special case when we have $\ell =1$ so the
$s$-vertex is the root of the trivial tree, see Fig. \ref{f_r2}.

\begin{figure}[!h]
	\centerline{\scalebox{0.7}{\includegraphics{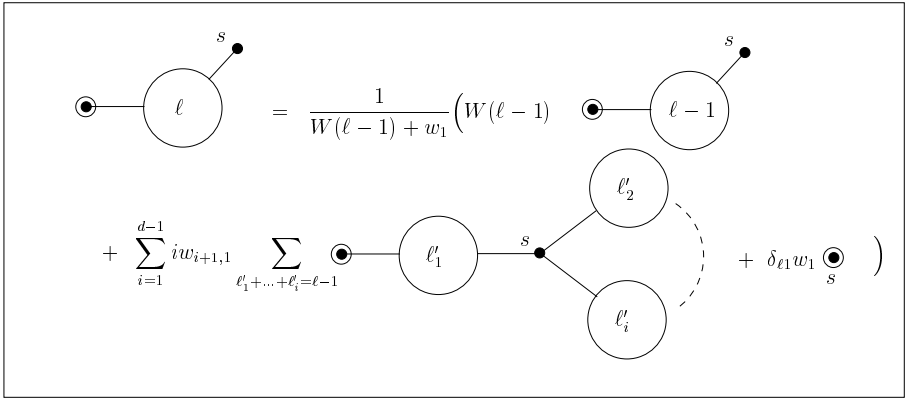}}}
\caption{A diagram representing equation (\ref{e_r3}).} \label{f_r2}
\end{figure}

For a general $k\geq 2$ and $s<\ell$ the recursion can be written
\begin{eqnarray} 
	\lefteqn{p_k(\ell_1,\ldots,\ell_k;s) \;\;=\;\; 
	\frac{1}{W(\ell-1)+w_1}\times }\nonumber \\ 
	&&\Big[\delta_{k2}\delta_{\ell_11}w_1p_R(\ell-1;s) 
	 + \sum_{i=1}^kW(\ell_i-1)
	p_k(\ell_1,\ldots,\ell_i-1,\ldots,\ell_k;s)  \label{e_r4} \\
	&&  + \sum_{i=k}^{d}\left(i+1-k\right)w_{k,i-k+2}
	\sum_{\ell'_1 + \ldots + \ell'_{i+1-k} = \ell_1-1} 
	p_i(\ell'_1,\ldots,\ell'_{i+1-k},\ell_2,\ldots,\ell_k;s) 
	\Big], \nonumber
\end{eqnarray}
see Fig.\ref{f_r3}. 
\begin{figure}[!h]
	\centerline{\scalebox{0.7}{\includegraphics{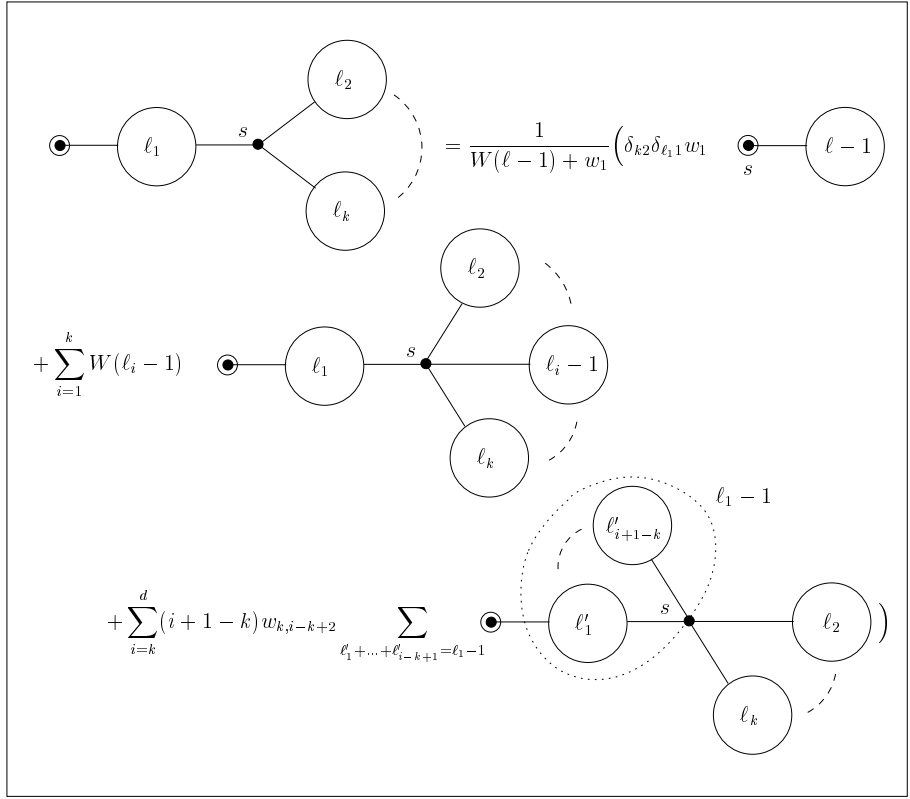}}}
\caption{A diagram representing equation (\ref{e_r4}).} \label{f_r3}
\end{figure}
The first term corresponds to the case when the
$s$-vertex is the root before the splitting in which case we have
$\ell_1=1$ and $k=2$.  The second term corresponds to the case when
we split a vertex different from the $s$-vertex and the last term
arises when we split the $s$-vertex in the step from time $\ell-1$ to
time $\ell$.  Finally we have 
\begin{eqnarray} \label{ee_r7}
       \lefteqn{p_k(\ell_1,\ldots,\ell_k;\ell) \;\;=\;\; 
	\dfrac{1}{W(\ell-1)+w_1} \times} \label{e_r5} \\ 
	&& \ds\sum_{s=0}^{\ell-1}\sum_{j=2}^{k}\sum_{i=k-1}^{d-1}
	\sum_{\ell'_1+\ldots + \ell'_{i+1-k} \atop = \ell_j-1} 
	w_{k,i-k+2} p_i(\ell_1,\ldots,\ell_{j-1},\ell'_1,\ldots,
	 \ell'_{i+1-k},\ell_{j+1},\ldots,\ell_k;s),  \nonumber
\end{eqnarray}
where $\ell_1 + \ldots + \ell_k = \ell$, see Fig.\ref{f_r4}.  
\begin{figure}[!h]
	\centerline{\scalebox{0.7}{\includegraphics{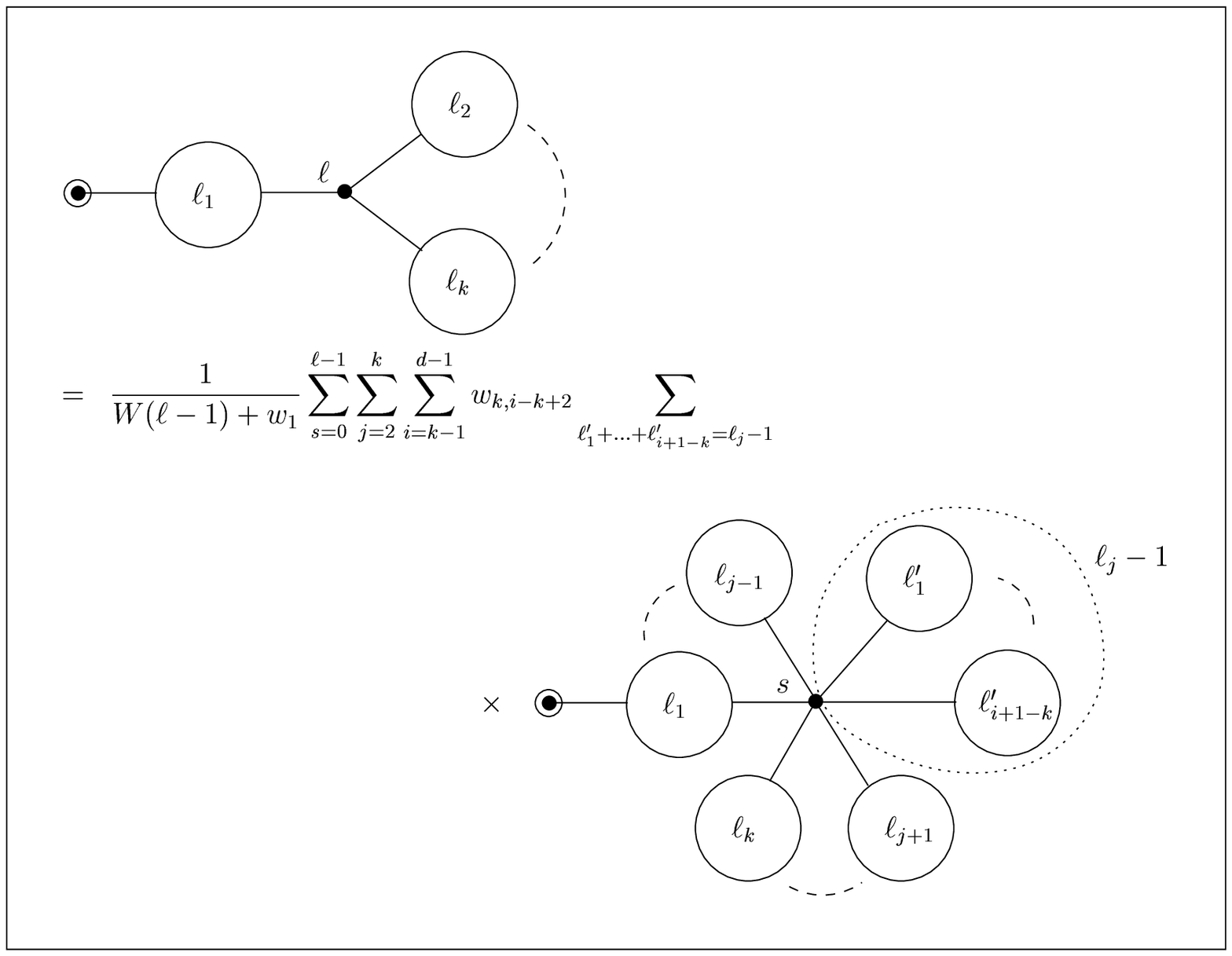}}}
\caption{A diagram representing equation (\ref{e_r5}).} \label{f_r4}
\end{figure}
Here $s$ is
the label of the vertex that is split in the step from time $\ell-1$
to time $\ell$ and we sum over all possible degrees of the $s$-vertex
and all ways of splitting it.

We define the following mean probabilities by averaging over the vertex labels in (\ref{e_r1}--\ref{ee_r7})
\begin{equation}
	 p_R(\ell) =\frac{1}{\ell+1}\sum_{s=0}^\ell p_R(\ell;s) 
	 \label{avroot}
\end{equation}
and
\begin{equation}
	 p_k(\ell_1,\ldots,\ell_k) = 
	\frac{1}{\ell+1}\sum_{s=0}^\ell p_k(\ell_1,\ldots,\ell_k;s), 
	\label{avpk}
\end{equation}
where $\ell_1 + \ldots + \ell_k = \ell$
From \rf{avroot} we get a recursion for the mean 
probabilities, going from time $\ell$ to $\ell+1$
\begin{eqnarray}
	 p_R(\ell+1) &=& \frac{\ell+1}{\ell+2}p_R(\ell)\label{recroot}. 
\end{eqnarray}
For $k=1$ we obtain from \rf{e_r23}, \rf{e_r3} and \rf{avpk}
\begin{eqnarray}
	\lefteqn{p_1(\ell+1)} \label{reck=1}\\ 
 	&=& 
	\frac{\ell+1}{\ell+2}\frac{1}{W(\ell)+w_1}\Big[W(\ell)p_1(\ell) 
	+\sum_{i=1}^{d-1}iw_{i+1,1}
	\sum_{\ell'_1+\ldots+\ell'_i \atop = \ell}
	p_i(\ell'_1,...,\ell'_i)+2\delta_{\ell0}w_1\Big].\nonumber 
\end{eqnarray}
Finally, the general case for $k\geq 2$ is
\begin{eqnarray}
	\lefteqn{p_k(\ell_1,\ldots,\ell_k)} \nonumber\\ 
	&=& \frac{\ell+1}{\ell+2}\frac{1}{W(\ell)+w_1}
	\Big[\delta_{k2}\delta_{\ell_11}w_1p_R(\ell) 
	+ \sum_{i=1}^kW(\ell_i-1)
	p_k(\ell_1,\ldots,\ell_i-1,\ldots,\ell_k) \nonumber\\
	&& + \sum_{i=k}^{d}\left(i-k+1\right)w_{k,i-k+2}
	\sum_{\ell'_1+\ldots+\ell'_{i+1-k} \atop = \ell_1-1}
	p_i(\ell'_1,\ldots,\ell'_{i+1-k},\ell_2,\ldots,\ell_k) 
	\label{reck}\\
	&&+ \sum_{j=2}^{k}\sum_{i=k-1}^{d}w_{k,i-k+2}
	\sum_{\ell'_1+\ldots+\ell'_{i+1-k}\atop = \ell_j-1}
	p_i(\ell_1,\ldots,\ell_{j-1},\ell'_1,\ldots,\ell'_{i+1-k},
		\ell_{j+1},\ldots,\ell_k)\Big] \nonumber
\end{eqnarray}
where $\ell_1 + \ldots + \ell_k = \ell+1$ and we have made use of 
\rf{e_r4}, \rf{e_r5} and \rf{avpk}.

\subsection{Two-point functions}
\label{ss2PtFunc}
One can reduce the above recursion formulas for the mean probabilities 
to simpler recursion formulas which suffice for the determination of
the Hausdorff dimension.   Define the {\it two-point functions}
\beq{twopoint}
	q_{ki}(\ell_1,\ell_2) 
	=	
	\sum_{\ell'_1 + \ldots + \ell'_{k-i} = \ell_1}
	\sum_{\ell''_1 + \ldots + \ell''_i = \ell_2} 
	p_k(\ell'_1,\ldots,\ell'_{k-i},\ell''_1,\ldots,\ell''_i),
\eeq
where $k=2,\ldots,d$ and $i = 1,\ldots,k-1$. 
In total there are $d(d-1)/2$ of these functions. 
If we define $$q_{1,0}(\ell_1,\ell_2) = \delta_{\ell_2 0}\delta_{\ell_1 \ell} p_1(\ell_1+\ell_2)$$ then $q_{ki}(\ell_1,\ell_2)$ is the probability that $i$ right trees of total volume $\ell_2$ are attached to a vertex of degree $k$ in a tree of total volume $\ell_1 + \ell_2$.
By summing over the equations in the previous section we get
\begin{eqnarray}
	q_{ki}(\ell_1,\ell_2)\nonumber
	&=& 
	\frac{\ell+1}{\ell+2}\; \frac{1}{W(\ell)+w_1}\Big[ \\
	&& \sum_{j=k-1}^d w_{k,j+2-k}
	\Big((j-i)q_{ji}(\ell_1-1,\ell_2)
	 +iq_{j,j-(k-i)}(\ell_1,\ell_2-1)\Big)\nonumber\\
	&& + \Big(W(\ell_1-1)+(k-i-1)(w_2-w_3)\Big)
		q_{ki}(\ell_1-1,\ell_2) \nonumber\\
	&& + \Big(W(\ell_2-1)+(i-1)(w_2-w_3)\Big)
		q_{ki}(\ell_1,\ell_2-1)\nonumber \\
	&& + \delta_{k2}\delta_{\ell_11}w_1p_R(\ell_2) 
	   + \delta_{i1}\delta_{\ell_21}w_{k,1}
	   \sum_{\ell'_1+\ldots+\ell'_{k-1} = \ell_1} 
	   p_{k-1}(\ell'_1,\ldots,\ell'_{k-1}) \Big] \nonumber \\ \label{2rec}
\end{eqnarray}
 with $\ell_1 + \ell_2 = \ell+1$.  We see that the two-point functions
 satisfy an essentially closed system of equations.  The last two terms in 
 \rf{2rec} do not contribute to 
 the scaling limit which will be discussed in the
 next section.

\section{Hausdorff dimension} 
\label{s:Hausdorff}

In this section we relate the two-point functions defined in the
previous section to the size of trees, defined in a suitable way.
With the help of some scaling assumptions this relation
allows us to calculate the Hausdorff dimension of the trees 
as a function of the partitioning weights in simple cases. As in Section \ref{s:subtree} we assume that the weights $w_i$ are linear in $i$ and we shall comment on the general case at the end of the section.

\subsection{Definition}
\label{ssDefHausdff}

Let $T$ be a tree with $\ell$ edges and $v$ and $w$ two vertices of
$T$. The (intrinsic geodesic) distance  $d_T(v,w)$ between $v$ and
$w$ is the number of edges that separate $v$ from $w$. We define the radius
of $T$ with respect to the vertex $v$ as
\begin{equation}
\label{Rv}
R_T(v)={1\over 2\ell}\sum_{w}d_T(v,w)\, k(w),
\end{equation}
where $k(w)$ is the degree of the vertex $w$. Notice that
$2\ell=\sum_w k(w)$.
The global radius of $T$ is
\begin{equation}
\label{Rglobal}
R_T={1\over 2\ell}\sum_{v}R_T(v) k(v).
\end{equation}
We define the Hausdorff dimension of the tree, $d_H$, from the scaling
law for large trees
\begin{equation}
\label{Rofell}
\langle R_T\rangle\ \sim\ \ell^{1/ d_H}\qquad{\ell\to\infty}.
\end{equation}
The more usual definition of the local fractal dimension of a vertex $v$ of the
tree, $d_f(v)$, is defined by the growth of the volume of a ball of radius $r$ around
$v$, $\mathcal{B}_v(r)=\{w:\,d(v,w)\le r\}$
\begin{equation}
\label{scalingdf}
\langle\mathrm{Card}(\mathcal{B}_v(r))\rangle\sim r^{d_f(v)}
\qquad  1\ll r\ll\ell.
\end{equation}
These two definitions are expected to coincide provided that the tree is a homogeneous fractal (no multifractal
behaviour).
We expect that the scaling behaviour \rf{scalingdf} 
is valid independently of
the point $v$ and also that the scaling (\ref{Rofell}) holds for $R_T(v)$
defined by (\ref{Rv}) irrespective of the point $v$ chosen.

\subsection{Geodesic distances and 2 point functions}
\label{ssGeodDist}
The radius $R_T(v)$ can be extracted from the two point functions
calculated in the previous section. Let $T$ be a tree and $v$ a vertex of $T$. 
Let $i$ be an edge of $T$. If we cut this edge
then the tree is split into two connected components,
a tree $T_1$ which contains $v$ and a tree $T_2$ that
does not contain $v$ (see Figure 8). 
Let $\ell_2(v;i)$ be the number of edges of $T_2$.
We have the simple remarkable result

\begin {equation} \label{RL}
\sum_w d_T(v,w) k(w) \;=\; \sum_i (2\ell_2(v;i)+1)
\end {equation}
\begin{figure}[h]
\begin{center}
\includegraphics[width=3in]{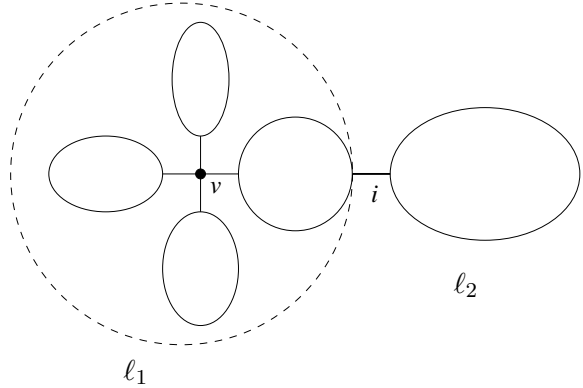}
\break
\hskip 3em\raisebox{0.em}{$\ell_1$}\hskip
1.6in\raisebox{7.ex}[0ex]{$\ell_2$}
\caption{Cutting a tree along the edge $i$.}
\label{f:treecut}
\end{center}
\end{figure}
which we will now prove. For the tree $T$ with $\ell$ edges, we may assign two labels to
every edge in the following way. 
Starting from $v$, we walk around the tree while 
always keeping the tree to the left.
Drop the labels 1 to $2\ell$ on the sides of edges as we pass them.

\begin{figure}[!h]
\includegraphics{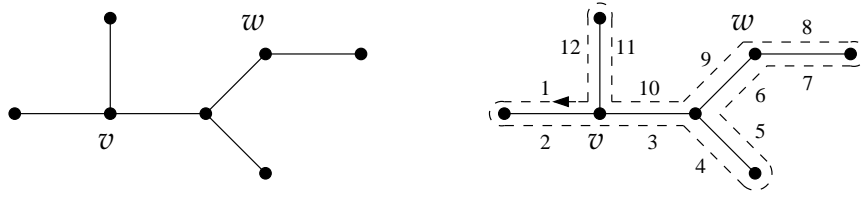}
\caption{A tree and its labels.}
\label{tree_label1}
\end{figure}

An example of such a walk and labelling is shown in Figure \ref{tree_label1}.
Let us mention that the initial direction from $v$ is unimportant.
In what follows we will denote these new labels by greek letters.

Given $1\leq \alpha <\beta\leq 2\ell$, define $\phi_v(\alpha,\beta)$ to 
be 1 if $\alpha$ and $\beta$ are labels of the same edge, and zero otherwise.
In the above example we have $\phi_v(6,9)=1$ whereas $\phi_v(6,12)=0$.
For any vertex $w \in T$, let us define $\omega(w)$ to be the smallest
label of the edges adjacent to $w$. In the example above $\omega(w)=6$
and $\omega(v)=1$.

We now have for any $w \in T$
\begin{eqnarray}\label{h&phi}
d_T(v,w) &=& \sum_{\alpha,\beta: \, \alpha\leq \omega(w) <\beta} \phi_v(\alpha,\beta)
\end{eqnarray}
and it is easy to see that
\begin{eqnarray} \nonumber
\sum_{w \in T} d_T(v,w) k(w)
  &=& 
  \sum_{\alpha,\beta,\gamma: \, \alpha\leq \gamma<\beta} \phi_v(\alpha,\beta) 
  \;=\;
  \sum_{\alpha,\beta: \, \alpha<\beta} \phi_v(\alpha,\beta) (\beta-\alpha).\\ \label{identity}
\end{eqnarray}

If $\phi(\alpha,\beta)=1$, i.e.\ if $\alpha$ and $\beta$
correspond to the two faces of the edge $a$, then
\begin{equation}
\label{q111}
\beta-\alpha= 2\,\ell_2(v,a)+1
\end{equation}
and equation (\ref{RL}) follows.

We now apply (\ref{RL}) by choosing for $v$ the root $r$ of the tree
and averaging over all trees obtained by the splitting process. We
notice that the average number of links giving the volume $\ell_2(r,i)$ is
simply the number of links $\ell$ times the proportion of vertices $w$
which have a left tree (containing the root $r$) on $\ell_1=\ell-\ell_2$
edges and an arbitrary number of right trees (on a total number of
$\ell_2$ edges). We obtain
\begin{equation}
\label{landqa}
\langle R_T(r)\rangle={\ell+1\over 2\ell}\sum_{\ell_2=0}^\ell(2\ell_2+1)
\sum_{k=1}^{d} q_{k,k-1}(\ell-\ell_2,\ell_2),
\end{equation}
see Fig.\ \ref{L&q} for a  simple diagrammatic representation of this
identity.
Thus, if we know how the two point functions $q_{k,i}(\ell_1,\ell_2)$
scale for large $\ell$, we know how the radius of the tree scales with
$\ell$ and we can compute the Hausdorff dimension $d_H$.
\begin{figure}
\begin{center}
\includegraphics[width=3in]{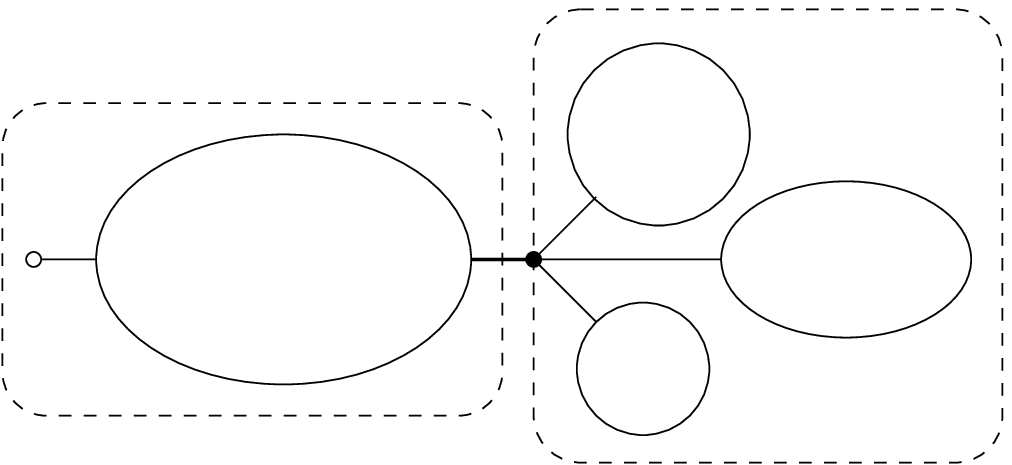}
\break
$\ell_1=\ell-\ell_2$\hskip1.2in$\ell_2$
\caption{Illustration of eq. \ref{landqa}.}
\label{L&q}
\end{center}
\end{figure}

\subsection{Scaling and the Hausdorff dimension}
\label{ScalingDF}

We assume that the following scaling holds for the two-point functions
$q_{ki}$ with $\ell$ large
\begin{eqnarray} \label{scalingassumption}
	q_{ki}(\ell_1,\ell_2) 
	&=& \ell^{-\rho}\left(\overline{\omega}_{ki}(x) + 
	\overline{\gamma}_{ki}(x)\ell^{-1}+O(\ell^{-2})\right)
\end{eqnarray}
where $\ell_1 + \ell_2 = \ell$, $x = \ell_1/\ell \in ]0,1[$ and where
$\overline{\omega}_{ki}, \overline{\gamma}_{ki}$ are some functions. It must hold that $\overline{\omega}_{ki} > 0$ and we assume that the scaling
exponent $\rho$ satisfies
\begin{equation}
\label{1rho2}
1< \rho\le 2 .
\end{equation}
Note that for $\ell$ finite, the probabilities
$q_{k,i}(\ell_1,\ell_2)$ are of order
$\ell^{-1}$ when
$\ell_1$ is of order 1 and are of
order $1$ when $\ell_2$ is of order $1$.
This implies that the scaling functions $\overline\omega_{ki}(x)$
should scale when $x\to 0$ or $x\to 1$,
respectively, as
\begin{equation}
\label{bbomega}
\overline\omega_{ki}(x)\mathop{\sim}_{x\to 0} x^{1-\rho}
\quad\text{and}\qquad
\overline\omega_{ki}(x)\mathop{\sim}_{x\to 1} (1-x)^{-\rho}.
\end{equation}
Using this ansatz and \rf{landqa} the mean radius scales as
\begin{equation}
\label{Rrscal}
\langle R_T(r)\rangle\simeq \ell^{2-\rho}\,C,\qquad
C=\int_0^1 dx\,
(1-x)\,\overline\omega(x),\qquad\overline\omega(x)=\sum_k \overline\omega_{k,k-1}(x).
\ \end{equation}
Equations (\ref{bbomega}) and (\ref{1rho2}) ensure that the integral $C$ is
convergent when $\rho <2$.  Equation
(\ref{Rofell}) then implies that the Hausdorff dimension of the tree is given
by
\begin{equation}
\label{rhodH}
2-\rho={1\over d_H}.
\end{equation}
For $\rho =2$ we see that $C$ is logarithmically divergent and this
corresponds to an infinite Hausdorff dimension.

Inserting (\ref{scalingassumption}) into the recursion equation (\ref{2rec}) for the two point functions and expanding in 
$\ell^{-1}$ gives (dropping the function argument $x$ in an obvious way)
\begin{eqnarray} \nonumber
	\lefteqn{
	\overline{\omega}_{ki} - \rho \overline{\omega}_{ki}\ell^{-1} 
	+ \overline{\gamma}_{ki}\ell^{-1} + O(\ell^{-2}) 
	} \\[0.5em] \nonumber
 	&=& \frac{1}{w_2}\ell^{-1}\Big(1-\frac{w_1+2w_2-w_3}{w_2}
	 \ell^{-1} + O(\ell^{-2})\Big) \\ \nonumber
 	&& \times \Big[\sum_{j=k-1}^d w_{k,j+2-k}\Big((j-i)
	 \overline{\omega}_{ji} +i\overline{\omega}_{j,j-(k-i)} 
	 + O(\ell^{-1})\Big) \\ \nonumber
 	&& + \ell\Big(w_2x+(-w_3+(k-i-1)(w_2-w_3))\ell^{-1}\Big)
	\Big(\overline{\omega}_{ki} + \overline{\gamma}_{ki}\ell^{-1} 
	+ O(\ell^{-2})\Big) \\
 	&& + \ell\Big(w_2(1-x)+i(w_2-w_3)\ell^{-1}\Big)
	 \Big(\overline{\omega}_{ki} + \overline{\gamma}_{ki}\ell^{-1} 
	 + O(\ell^{-2})\Big)\Big].
\end{eqnarray}
The equation is trivially satisfied in zeroth order of $\ell^{-1}$. 
When we go to the next order we see that the following must hold
\begin{eqnarray} \nonumber
	(2-\rho) \overline{\omega}_{ki} 
	&=& 
	\frac{1}{w_2}\sum_{j=k-1}^d w_{k,j+2-k}\Big((j-i)
	 \overline{\omega}_{ji} +i\overline{\omega}_{j,j-(k-i)}\Big) 
	 - \frac{w_k}{w_2}\overline{\omega}_{ki}. \\ \label{mainII}
\end{eqnarray}
This eigenvalue equation may be rewritten as 
\begin{eqnarray} \label{eighaus}
\mathbf{C}\mbox{\boldmath $\omega$}&=&w_2(2-\rho)\mbox{\boldmath $\omega$}
\end{eqnarray}
\\
where $\mathbf{C}$ is a \small{$\binom{d}{2} \times \binom{d}{2}$} \normalsize matrix indexed by a pair of two indices $ki$ with $k>i, ~k=2,\ldots,d$ and \mbox{\boldmath $\omega$} is 
a vector with two such indices. The matrix elements of $\mathbf{C}$ are 
\begin{eqnarray} \label{Ccoeff}
	C_{ki,jn} &=& 
	w_{k,j+2-k}\left(\left(j-i\right)\delta_{in} 
	 + i\delta_{n,j-(k-i)}\right) -w_k \delta_{kj}\delta_{in}.
\end{eqnarray}
We use the convention that $w_{i,j} = 0$ if $i$ or $j$ is less than 1 or greater than $d$.
Thus, $w_2(2-\rho)$ is an eigenvalue of the matrix $\mathbf{C}$ and the associated eigenvector must have components $\ge 0$. We now show that there is in general a unique solution to this eigenvalue problem. 

Since the only possibly negative elements of $\mathbf{C}$ are on the diagonal we can make the matrix non-negative by adding a positive multiple $\gamma$ of the identity to both sides of (\ref{eighaus}) and choosing $\gamma$ large enough.

If enough of the weights $w_{i,j}$ are non--zero ($w_{1,i} > 0$ for $2 \leq i \leq d$ and $w_{j,3} > 0$ for $2 \leq j \leq d-1$ is for example sufficient) then one can check that the matrix $\mathbf{C}+\gamma\mathbf{I}$ is primitive. Then, by the Perron-Frobenius theorem, it has a simple positive eigenvalue of largest modulus and its corresponding
eigenvector can be taken to have all entries positive cf. Lemma \ref{perronfrob}. 
Therefore this largest positive eigenvalue gives the $\rho$ we are after.

\subsection{An upper bound on the Hausdorff dimension}
\label{ssDFbound}

We can get an upper bound on $\rho$ by a straight forward estimate 
from (\ref{mainII}). The off-diagonal terms in the sum are all 
non-negative so we disregard them and get the inequality

\begin{eqnarray}
	\rho & \leq & 
	2-\left(k\frac{w_{k,2}}{w_{2}}-\frac{w_k}{w_2}\right), \quad k=2,\ldots,d.
\end{eqnarray}
Since $1 < \rho \leq 2$ and for $k \geq 3$
\begin{eqnarray}
	w_k &=& 
	k w_{k,2} + \frac{k}{2}
	\sum_{\substack{i=1\\i\neq 2,~i\neq k}}^{k+1}
		w_{i,k+2-i} \;>\; kw_{k,2}
\end{eqnarray}
the best we can get from this upper bound is when $k = 2$ which yields

\begin{eqnarray} \label{uprho}
 	\rho &\leq& 2- \frac{w_{2,2}-2w_{1,3}}{w_{2,2}+2w_{1,3}}.
\end{eqnarray}
Now, $2-\rho = \frac{1}{d_H}$ and therefore, if $w_{2,2} > 2w_{1,3}$, 
we obtain the upper bound

\begin{eqnarray} \label{uphaus}
	d_H &\leq&  \frac{w_{2,2}+2w_{1,3}}{w_{2,2}-2w_{1,3}}.
\end{eqnarray}
If $w_{2,2} \leq 2w_{1,3}$ the upper bound in (\ref{uprho}) gives no information about the Hausdorff dimension.

It is easy to verify that (\ref{uphaus}) is an equality if we choose the weights 
such that $w_{i,j} = 0$ if $i\neq 1$ or $j\neq 1$ with the exception 
that $w_{2,2} > 0$.  This condition means that we only allow vertices to 
evolve by link attachment, except that we can split vertices of degree 
2.  With this choice the matrix $\mathbf{C}$ is lower 
triangular and we can simply read the eigenvalues from the diagonal. Note that $\mathbf{C}$ is not primitive in this case and therefore we cannot use the Perron-Frobenius theorem to determine which eigenvalue gives the scaling exponent. However, with these simple weights one can show explicitly that there is precisely one eigenvector with strictly positive components and the corresponding eigenvalue is the one that saturates the inequality in (\ref{uphaus}).  

Also note that with this choice we have set $w_d=0$ and since the 
weights are linear, $w_i = ai+b$, we have fixed $a$ and $b$ so that 
$w_i = 1-\frac{i}{d}$. Therefore there is only one free parameter which we can choose to be $w_{2,2}$. Then we can write the Hausdorff dimension as
\begin {equation}
 d_H = \frac{1-2/d}{2w_{2,2}-(1-2/d)} 
\end {equation}
with
\begin {equation}
  \quad \frac{1}{2}(1-2/d) < w_{2,2} < 1-2/d.
\end {equation}
We see that for any $d$ the Hausdorff dimension can vary continuously from 1 to infinity.

\subsection{Explicit solutions and numerical results for $d_{\mathrm{max}}=3$.}
\label{ss=expd=3}
When the maximal degree is $d=3$, the splitting weights are taken to be linear \\ $w_i = ai+b$ and the partitioning weights uniform, it is easy to solve equation (\ref{mainII}) for the Hausdorff dimension . Since the solution only depends on the ratio of the weights there is only one independent variable and we choose it to be $y := w_3/w_2$ where $0 \leq y \leq 2$. The solution is  

\begin{equation} \label{specialHaus}
	d_H = \frac{3(1+\sqrt{1+16y})}{8y}
\end{equation}

In Figure \ref{f:hausdorff} we compare this equation to 
results from simulations. The agreement of the simulations with the formula is good in the tested range $0.5 \leq y \leq 2$. For smaller values of $y$ the Hausdorff dimension increases fast and one would have to simulate very large trees to see the scaling.

\begin{figure}[h]
\centerline{{\scalebox{0.8}{\includegraphics{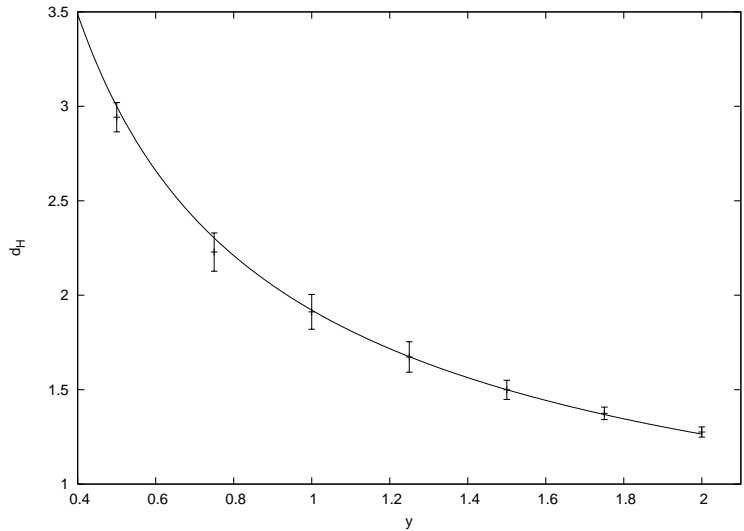}}}} 
\caption{Equation (\ref{specialHaus}) compared to simulations. 
The Hausdorff dimension, $d_H$, is plotted against $y = w_3/w_2$. The leftmost datapoint is calculated from 50 trees on 50000 vertices and the others are calculated from 50 trees on 10000 vertices.} \label{f:hausdorff}
\end{figure}

\subsection{Hausdorff dimension for general weights}
\label{ss:HausGen}
\subsubsection{General mean field argument}
\label{sss:HausMF}
Our argument to compute the Hausdorff dimension relies on the recursion relations for the substructure probabilities, studied in Section~3, which are valid only when the splitting weights $w_i$ are linear functions of the vertex degree $i$ ($w_i=ai+b$). In this case the total probability weight $\mathcal{W}(T)$ for a given tree $T$ depends only on its size $\ell$ (number of edges) and mean field arguments can be made exact.

In the general case where the $w_{i,j}$ are arbitrary and the $w_i$ are not linear with $i$, these recursion relations are no more exact.
We can use a mean field argument and assume that they are still valid for large ``typical'' trees, provided that we replace in these recursion relations the exact weights $W(\ell)+w_1=\mathcal{W}(T)$ by their mean field value for large trees 
\begin{equation}
\label{W2Wover}
W(\ell)+w_1\ \to\ \overline{\mathcal{W}}(\ell)=\sum_j w_j \overline{n}_{\ell+1, j}
\end{equation}
where $\overline{n}_{\ell+1, j}$ is the average number of $i$-vertices in a tree with $\ell$ edges, studied in Section~2. 
From the mean field analysis of Section~2.5, we expect that  these $\overline{n}_{\ell+1, j}$ scale with $\ell$ as
\begin{equation}
\label{W2WMF}
\overline{n}_{\ell+1, j}\simeq\ell\rho_j
\end{equation}
with the vertex densities $\rho_j$ given by the mean field equations (\ref{mfequation},\ref{tme2},\ref{tme3}) as the components of the eigenvector $\boldsymbol{\rho}$ associated to the largest eigenvalue $w$ of the matrix $\mathbf{B}$.
Thus the mean field approximation amounts to replacing 
\begin{equation}
\label{Wwell}
W(\ell)+w_1\ \to\ \overline{\mathcal{W}}(\ell)=w\, \ell+\cdots
\end{equation}
in the recursion relations of Section~3, in particular in the recursion relation (\ref{2rec}) for the two point function $q_{ki}$.

With this assumption, we can repeat the scaling argument of Section~4.3, and we end up with equation (4.16), with the normalisation factor ${1\over w_2}$ in the r.h.s. replaced by the mean field normalisation factor ${1\over w}$
\begin{equation} \label{main}
	(2-\rho)\, \overline{\omega}_{ki} 
	= 
	\frac{1}{w}\sum_{j=k-1}^d w_{k,j+2-k}\Big((j-i)
	 \overline{\omega}_{ji} +i\overline{\omega}_{j,j-(k-i)}\Big) 
	 - \frac{w_k}{w}\overline{\omega}_{ki}
	 \ .
\end{equation}
This equation is still an eigenvalue equation of the form
\begin{equation}
\label{evmain}
\mathbf{C}\,\boldsymbol{\omega} = w(2-\rho)\,\boldsymbol{\omega}
\end{equation}
where $\mathbf{C}$ is the $\binom{d}{2}\times\binom{d}{2}$ matrix with coefficients given in (\ref{Ccoeff}).

If we denote by $\chi$ the largest eigenvalue of this matrix $\mathbf{C}$ and if $w$ is the largest eigenvalue of $\mathbf{B}$
then the Perron-Frobenius argument can be applied to show that $\chi$ is non-negative and that the eigenvector $\boldsymbol{\omega}$ has non-negative components, which is a consistency requirement for the argument, since the $\omega_{ki}$ are rescaled probabilities.
We end up with a mean field prediction for the Hausdorff dimension of the simple form
\begin{equation}
\label{dHmf}
d_H={1\over 2-\rho}={w\over\chi}.
\end{equation}

\subsubsection{General solution for $d=3$}
\label{sss:dHd=3}
In the $d=3$ case (trees with only $k=1,2$ and $3$ vertices), the $\mathbf{B}$ and $\mathbf{C}$ matrices are 
\\
\begin{eqnarray}
\label{ B&Cd=3}
\mathbf{B}=\left[\begin{matrix}
     0 &  2 w_{3,1} & 0  \\
      w_{2,1}& w_{2,2}-2\,w_{3,1} & 3 \,w_{3,2}\\
    0   & 2\,w_{3,1} & 0
\end{matrix}\right],
\hskip0.1in
\mathbf{C}=\left[\begin{matrix}
     w_{2,2}-2\,w_{3,1} &  2 w_{3,2} & w_{3,2}  \\
      w_{3,1}& 0 & 0\\
    2\,w_{3,1}   & 0 & 0
\end{matrix}\right] \nonumber \\ \nonumber\\
\end{eqnarray}
and we find
\\
\begin{equation}
\label{dHd=3}
d_H=\frac{(w_{2,2}-2 w_{3,1})+\sqrt{(w_{2,2}-2 w_{3,1})^2+8 w_{3,1}(w_{2,1}+3 w_{3,2})}}{(w_{2,2}-2 w_{3,1})+\sqrt{(w_{2,2}-2 w_{3,1})^2+16 w_{3,1}w_{3,2}}}.
\end{equation}
\\
\subsubsection{Numerical test of the mean field prediction}\
\label{sss:dHtest}
We have tested our prediction when $d=3$ and the partitioning weights $w_{i,j}$ are uniform. In this case the \ref{dHd=3} becomes
\begin{equation} \label{generalHaus}
d_H = {\frac {\alpha-\sqrt {\alpha\, \left( \alpha+24+24\,\beta \right) }}{
\alpha-\sqrt {\alpha\, \left( \alpha+16\,\beta \right) }}}
\end{equation}
\\
where $\alpha = \dfrac{w_2}{w_1}$ and $\beta = \dfrac{w_3}{w_1}$. 
In Figure \ref{f:genhausdorff} we compare this equation to 
results from simulations. There is a good agreement for small values of $\alpha$ and not too small values of $\beta$, but the precision of the numerics becomes poor for large values of $\alpha$ and small values of $\beta$. This could be  expected since in this case, the trees will have a large Hausdorff dimension and finite size effects are expected to be large, so that one must go to very large trees to see the scaling. Clearly a better estimate of finite size effects and large simulations are desirable.
\begin{figure}[h]
\centerline{{\scalebox{0.8}{\includegraphics{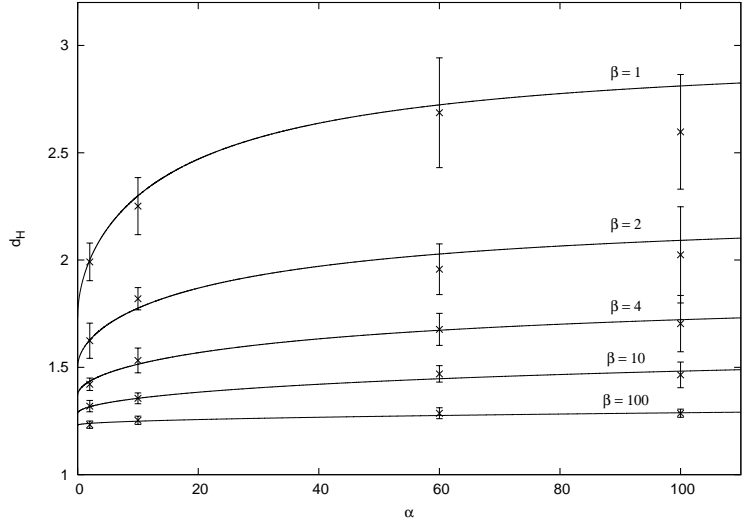}}}} 
\caption{Equation (\ref{generalHaus}) compared to simulations. 
Each datapoint is calculated from 50 trees on 10000 vertices.} \label{f:genhausdorff}
\end{figure}

\section{Correlations} \label{s:correlations}

\subsection{Vertex-vertex correlations in the linear case}
\label{ssVVlinear}

Consider a tree on $\ell$ edges generated by the splitting 
procedure starting from the single vertex tree at time 0. We are interested in determining the density of edges which have endpoints of degrees $j$ and $k$ in the limit when $\ell \rightarrow \infty$. Denote this density by $\rho_{jk}$, assuming it exists, and for convenience we let the vertex of degree $j$ be the one closer to the root. Therefore $\rho_{1k}=0$ for all $k$ and in general $\rho_{jk} \neq \rho_{kj}$. 

To arrive at these densities we use the same labelling techniques as in Ssection \ref{s:subtree}. We start by defining 
$$p_{jk}(\ell'_1,\ldots,\ell'_{j-1};\ell''_1,\ldots,\ell''_{k-1};s)$$ 
as the probability that a vertex created at time $s$ is of degree 
$k$ and has $\ell''_1,\ldots,\ell''_{k-1}$ right trees and that the 
vertex left to it is of degree $j$ with an $\ell'_1$ left tree and 
$\ell'_2,\ldots,\ell'_{j-1}$ right trees 
(excluding the right tree which contains $s$). Note that it is symmetric under permutations of both $(\ell'_2,\ldots,\ell'_{j-1})$ 
and $(\ell''_1,\ldots,\ell''_{k-1})$ and
\begin{eqnarray*}
\ell'_1+\ldots+\ell'_{j-1}+\ell''_1 + \ldots + \ell''_{k-1} &=& \ell-1. 
\end{eqnarray*}

We use the methods of Section \ref{s:subtree} to derive recursion equations for 
\\ $p_{jk}(\ell'_1,\ldots,\ell'_{j-1};\ell''_1,\ldots,\ell''_{k-1};s)$ for linear splitting weights. 
We then average over the label $s$ as before and get recursions for the 
average probabilities 
$p_{jk}(\ell'_1,\ldots,\ell'_{j-1};\ell''_1,\ldots,\ell''_{k-1})$. 
All equations and diagrams are given in Appendix \ref{A:corr}. Finally 
we define the densities $\rho_{jk}(\ell)$ by averaging 
out the volume dependence of the average probabilities
\begin{eqnarray*}
	\rho_{jk}(\ell) &=& 
	\sum_{\ell'_1+\ldots\ell'_{j-1}+\ell''_1+\ldots+\ell''_{k-1}
	=\ell-1}
	p_{jk}(\ell'_1,\ldots,\ell'_{j-1};\ell''_1,\ldots,\ell''_{k-1})
\end{eqnarray*}
and similarly we denote the vertex degree density by
\begin{eqnarray*}
	\rho_j(\ell) \equiv \rho_{\ell,j} &=& 
	\sum_{\ell_1+\ldots+\ell_j=\ell} p_j(\ell_1,\ldots,\ell_j),
\end{eqnarray*}
cf. Section \ref{s:genfunc}.
We have the following recursion for the densities
\begin{eqnarray*}
	\lefteqn{\rho_{jk}(\ell+1) \;\;=}\\ 
	&&\frac{\ell+1}{\ell+2}\frac{1}{W(\ell)+w_1}
	\bigg\{(\ell w_2-w_j-w_k+2w_1-w_2)\rho_{jk}(\ell) 
		+ (j-1)w_{j,k}\rho_{j+k-2}(\ell) \\
	&&+ \; (j-1)\sum_{i=j-1}^d w_{j,i+2-j}\rho_{ik}(\ell) 
		+ (k-1)\sum_{i=k-1}^dw_{k,i+2-k}\rho_{ji}(\ell) \\
	&&+ \; \delta_{j1}\delta_{\ell''_1,\ell-1}w_1p_R(\ell)
		+\delta_{j1}\delta_{\ell0}w_1\bigg\}
\end{eqnarray*}
for $i,j \geq 1$, see Appendix \ref{A:corr} . Now assume that
$\rho_{jk}(\ell) = \rho_{jk} + r_{jk}\ell^{-1} + O(\ell^{-2})$ 
and that a similar expansion holds for $\rho_j(\ell)$. 
Expanding the above recursion in $\ell^{-1}$ gives
\begin{eqnarray*}
\lefteqn{\rho_{jk}+r_{jk}\ell^{-1} + O(\ell^{-2}) \;\;=\;\;
\left(1-\frac{w_1+2w_2-w_3}{w_2}\ell^{-1}+O(\ell^{-2})\right) 
\bigg\{ } \\
&& \left(1+\frac{-w_j-w_k+2w_1-w_2}{w_2}\ell^{-1}\right)
	\left(\rho_{jk} + r_{jk}\ell^{-1}+O(\ell^{-2})\right) \\
&&+ \frac{\ell^{-1}}{w_2}
	\bigg[(j-1)w_{j,k}\big(\rho_{j+k-2}+O(\ell^{-1})\big)
\;+\; (j-1)\sum_{i=j-1}^d w_{j,i+2-j}
	\left(\rho_{ik}+O(\ell^{-1})\right) \\ 
&&+ (k-1)\sum_{i=k-1}^dw_{k,i+2-k}\left(\rho_{ji}+O(\ell^{-1})\right) 
	\bigg]\bigg\}.
\end{eqnarray*}
This equation is trivially satisfied in zeroth order of $\ell^{-1}$. 
When we go to the next order we get the following equation for the limits of the densities
\begin{eqnarray} 
	\rho_{jk} &=& 
	-\frac{w_j+w_k}{w_2} \rho_{jk} + (j-1) \frac{w_{j,k}}{w_2}
	\rho_{j+k-2}+(j-1)\sum_{i=j-1}^{d}
		\frac{w_{j,i+2-j}}{w_2}\rho_{ik} \nonumber\\ \label{correquation}
	&& + (k-1)\sum_{i=k-1}^d\frac{w_{k,i+2-k}}{w_2}\rho_{ji}. 
		\label{corrsol}
\end{eqnarray}

We can also arrive directly at this equation by assuming that for large 
$\ell$ an equilibrium with small enough fluctuations is established, 
and then perform the splitting procedure on this equilibrium. With the same methods, it is possible to derive an equation like (\ref{correquation}) for the density $\rho_{j_1,j_2,\ldots,j_R}$ of linear paths of length $R-1$ directed towards the root containing vertices of degrees $j_1,\ldots,j_R$.

Existence of solutions to equation (\ref{correquation}) can be established by the Perron-Frobenius argument like in the previous sections. In the following subsections we will find an explicit solution for linear weights and discuss generalizations for non-linear weights. In both cases we compare the results with simulations.

\subsection {Solution in the simplest case}
\label{ssVVd=3}
When $d=3$, the splitting weights are linear and the partitioning weights are uniform, we know that 
$\rho_1 = \rho_3 = 2/7$ and $\rho_2 = 3/7$, see Section \ref{s:explicitsolutions}.
Let $y=w_3/w_2$. Then the solutions to equation (\ref{corrsol})  are

\begin {equation} \label{corrsolutions}
\begin{array}{rcl@{\qquad}rcl}
\rho_{21} &=&\dfrac{4(3-y)}{7(11-2y)}, &
\rho_{31} &=& \dfrac{10}{7(11-2y)},\\[1.5em]
\rho_{22} &=& \dfrac{4y^2-12y+105}{7(2y+7)(11-2y)}, &
\rho_{32} &=& \dfrac{2(-8y^2+18y+63)}{7(2y+7)(11-2y)},\\[1.5em]
\rho_{23} &=&\dfrac{2(-4y^2+20y+21)}{7(2y+7)(11-2y)}, &
\rho_{33} &=& \dfrac{8(3y-14)}{7(2y+7)(2y-11)}.\\[1.5em]
\end{array}
\end {equation}

Note that the following sum rules hold for the solutions
\begin {equation} \label{sumrulescorr}
\begin{array}{rclcl}
\rho_{21} + \rho_{31} &=& \rho_1 &=& {2}/{7} \\
\rho_{22} + \rho_{32} &=& \rho_2 &=& {3}/{7} \\
\rho_{23} + \rho_{33} &=& \rho_3 &=& {2}/{7}, \\
\rho_{21} + \rho_{22} + \rho_{23} &=& \rho_2 &=& {3}/{7} \\
\rho_{31} + \rho_{32} + \rho_{33} &=& 2\rho_3 &=& {4}/{7}.
\end{array}
\end {equation}
These relations show that there are only two independent link densities. We plot $\rho_{21}$ and $\rho_{22}$ in Figure \ref{f_sol2} and 
compare to simulations.

\begin{figure}[h]
\centerline{\scalebox{0.8}{\includegraphics{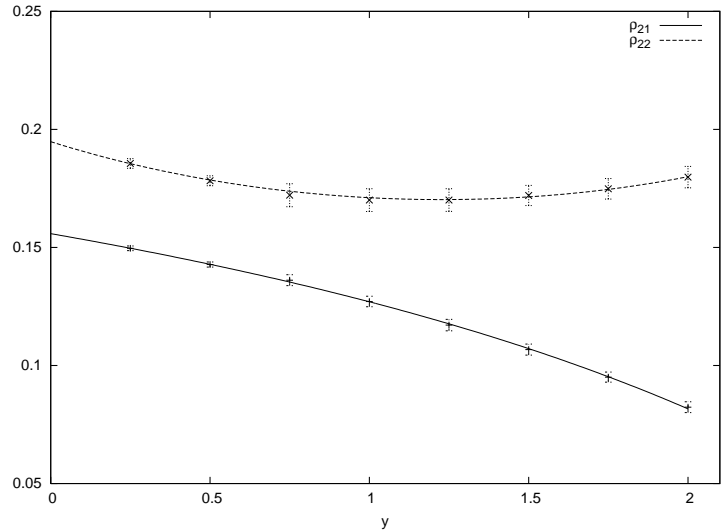}}} 
\caption{
Two independent solutions given in (\ref{corrsolutions}) plotted against $y = w_3/w_2$ and compared to simulations. The two leftmost datapoints on each line come from simulations of 50 trees on 50000 vertices. The other datapoints come from simulations of 50 trees on 10000 vertices.} \label{f_sol2}
\end{figure}

We can compare the above results to the case when no correlations are present. 
Denote the uncorrelated densities by $\tilde{\rho}_{ij}$. 
Then we simply have 
\begin{eqnarray*}
\tilde{\rho}_{ij} &=& \frac{\rho_i \rho_j}{1-\rho_1}.
\end{eqnarray*}
\\
The denominator comes from the fact that the vertex closer to the root 
is of degree one with probability zero. Inserting the values from (\ref{sumrulescorr}) into this equation 
gives

$$\begin{array}{rcl@{\qquad}rcl}
\tilde{\rho}_{21} &=& {6}/{35}, & \tilde{\rho}_{31} &=& {4}/{35}\\
\tilde{\rho}_{22} &=& {9}/{35}, & \tilde{\rho}_{32} &=& {6}/{35}\\
\tilde{\rho}_{23} &=& {6}/{35}, & \tilde{\rho}_{33} &=& {4}/{35}
\end{array}$$
\\
showing that in general $\rho_{ij} \neq \tilde{\rho}_{ij}$ and so correlations are present between degrees of vertices.

\subsection{Results for non-linear weights}
\label{VVnonlinear}
We can generalize equation (\ref{corrsol}) to a mean field equation, valid for 
arbitrary weights, by replacing $w_2$, where it appears in a denominator, with $w$ as we did with the equation for vertex degree densities in
Section \ref{S:genfunc}. For $d=3$ and uniform partitioning weights the two independent densities $\rho_{21}$ and $\rho_{22}$ are given by

\begin {equation}
\label{e:r21}
\rho_{21}=\frac{1}{3}\,{\frac { \left( 3+\beta \right)  \left( 7\,\alpha-\gamma \right) }{ \left(2
\,\alpha-\beta -1\right)  \left( 3\,\alpha+2\,\beta+\gamma +6\right) }}
\end {equation}
\\
\scriptsize
\begin {eqnarray*} \nonumber
\text{{\normalsize $\rho_{22}$}} &\text{{\normalsize$=$}}& \frac{16}{3}\Big(284\,{{\alpha}^{2}\beta}^{4}\gamma-177\,{\alpha}^{5}\beta\gamma+3564\,{\alpha}^{3}+18\,{\alpha}^{6}\gamma +161\alpha\,{\beta}^{5}\gamma-873\,\gamma+11979\,{\alpha}^{2}{\beta}^{3} \\ \nonumber
&&- 2259\,{\alpha}^{5}-39\,{\alpha}^{6}\beta-207\,{\alpha}^{5}\gamma+6516\,{\alpha}^{2}{\beta}^{4}-5205\,{\alpha}^{5}\beta-1419\,{\alpha}^{4}\beta\gamma+996\,\alpha{\beta}^{5}\\ \nonumber
&& - 5994\,{\alpha}^{4}-892\,{\alpha}^{4}{\beta}^{2}\gamma+1543\,{\alpha}^{2}{\beta}^{5}-18\,{\alpha}^{7}-668\,{\alpha}^{3}{\beta}^{4}+324\,{\alpha}^{2}\gamma+909\,{\alpha\beta}^{3}\gamma\\ \nonumber
&&-2600\,{\alpha}^{5}{\beta}^{2}-975\,{\alpha}^{3}{\beta}^{3}+222\,\alpha{\beta}^{6}-1533\,{\alpha}^{3}{\beta}^{2}\gamma+10206\,{\alpha}^{2}{\beta}^{2}-11799\,{\alpha}^{4}\beta\\ \nonumber
&&-5300\,{\alpha}^{4}{\beta}^{3}-1521\,{\alpha}^{3}\beta\gamma+1899\,{\alpha}^{2}{\beta}^{2}\gamma+1059{\alpha}^{2}\,{\beta}^{3}\gamma+1269\,{\alpha}^{3}{\beta}^{2}+3240{\alpha}^{2}\,\beta \\ \nonumber
&&+756\,\alpha{\beta}^{3}+4860\,{\alpha}^{3}\beta+6\,{\beta}^{6}\gamma-11703\,{\alpha}^{4}{\beta}^{2}+1728{\alpha}^{2}\,\beta\gamma-162\,{\alpha}^{3}\gamma+486\alpha\,{\beta}^{2}\gamma\\ \nonumber
&&+18\,{\beta}^{4}\gamma+1530\,\alpha{\beta}^{4}+
624\alpha\,{\beta}^{4}\gamma-772\,{\alpha}^{3}{\beta}^{3}\gamma-9\,{\alpha}^{6}+24\,{\beta}^{5}\gamma\Big)\Big/\Big(\left( 3\,\alpha+2\,\beta+\gamma+6
 \right) \\ \nonumber
 &&\times\left( 11\,{\alpha}^{2}+25\,\alpha\beta+5\,\alpha\gamma+3\,\beta\gamma+12\,\alpha+4\,{\beta}^{2} \right) \left( -\alpha+\gamma \right)  \left( 1-2\,\alpha+\beta \right)    \left( 7\,\alpha+2\,\beta+\gamma \right) ^{2}
\Big)  \\ 
\end {eqnarray*}
\normalsize
\begin {equation}  \label{e:r22}
~
\end {equation}
\\
\normalsize
where $\alpha = \dfrac{w_2}{w_1}$, $\beta = \dfrac{w_3}{w_1}$ and $\gamma=\sqrt {\alpha \left(\alpha+24\,\beta +24\right) }$.  These solutions are compared to simulations in Figures \ref{f_rho21} and \ref{f_rho22}. The other densities are obtained by using the sum rules (\ref{sumrulescorr}).

\begin{figure}[!h]
\centerline{\scalebox{0.9}{\includegraphics{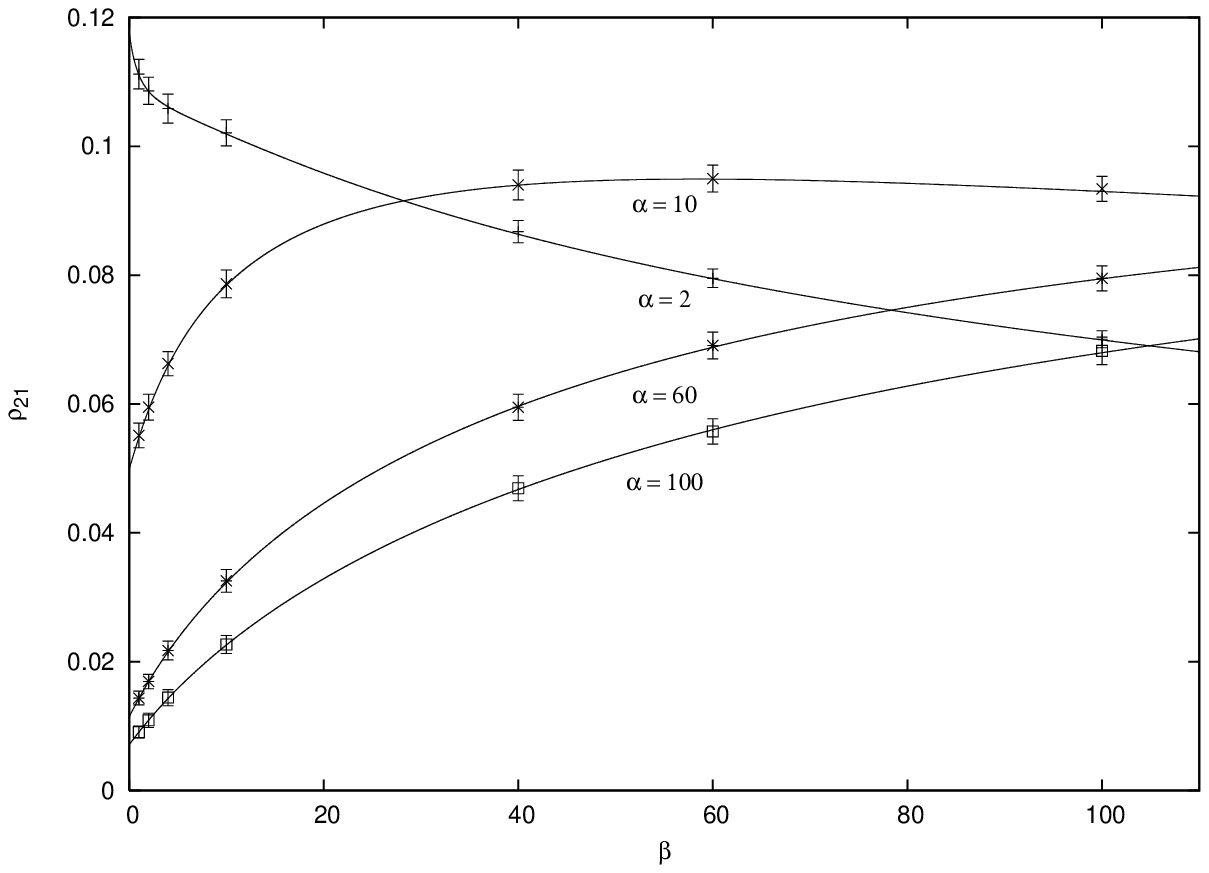}}} 
\caption{
The  solution (\ref{e:r21}) for the density $\rho_{21}$ plotted as a function of $\beta$ for a few values of $\alpha$. Each datapoint is calculated from simulations of 100 trees on 10000 vertices.} \label{f_rho21}
\end{figure}
\begin{figure}[!h]
\centerline{\scalebox{0.9}{\includegraphics{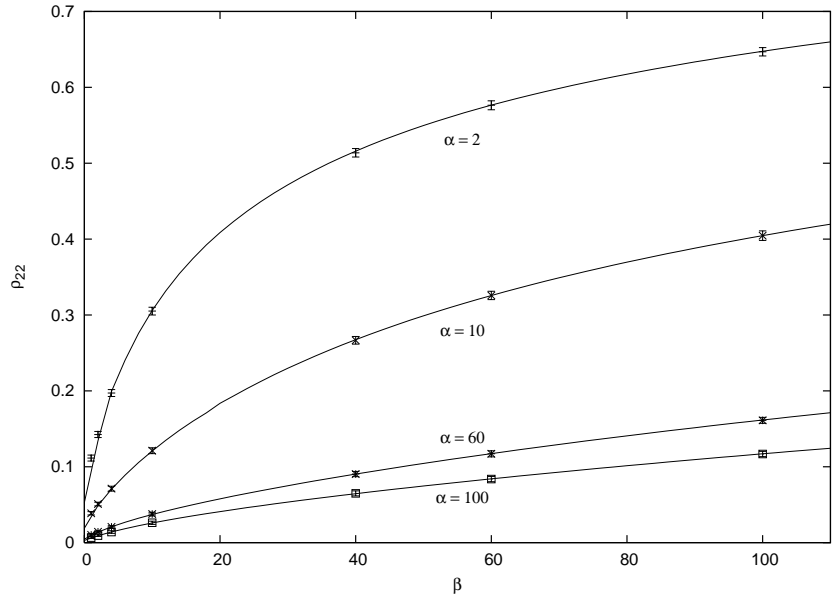}}} 
\caption{
The solution  (\ref{e:r22}) for the density $\rho_{22}$ plotted as a function of $\beta$ for a few values of $\alpha$. Each datapoint is calculated from simulations of 100 trees on 10000 vertices.} \label{f_rho22}
\end{figure}

\section{Discussion}
\label{s:Disc}
In this paper we introduced a simple but general model of random tree growth by vertex splitting. 
The probability weights $w_{i,j}$ for each splitting process can be chosen independently and are the parameters of the model.
We study the large time/large tree size asymptotics of this model.
A mean field theory is presented for the vertex degree distributions and we deduce master equations for various correlations functions, like subtree structure probabilities and vertex degree correlations. 
The mean field results are proved in some special cases and checked by numerical simulations in many general cases.
Under a natural scaling hypothesis we are able to compute the intrinsic Hausdorff dimension $d_H$ of the tree.
This Hausdorff dimension may take values between $d_H=1$ and $d_H=\infty$, depending on the parameters $w_{i,j}$.
These predictions for $d_H$ are checked by numerical simulations.

There are still many open questions for the splitting model presented here. 
The most urgent one is to prove rigorously that the mean field theory is valid in the general case where the splitting weights $w_i$ are not linear with the vertex degrees $i$. 
Another goal is to get a better understanding of the degree-degree correlations for vertices which are at a finite distance greater than one, and of the general substructure probabilities. 
It would also be very interesting to obtain general distribution functions for correlations involving the geodesic distances on the trees, and more than two points.

\medskip
As already mentioned in the introduction, the model considered here is related to other tree growth models. 
Firstly, it contains as a very special subclass the so called preferential attachment models, which have been extensively studied in the literature (see e.g.~\cite{MASS_DIST,AlbertBara}).
This model  consists in allowing only growth by attaching  
a new vertex by an edge to an already existing $k$-vertex,
with an attachment rate $w_k$ depending only on the degree $k$ of the
initial vertex.  In our model this corresponds to allowing only the
splitting processes
$$k\to 1,k+1\ ,$$
and to set all the splitting weights to zero except the 
$$w_{k+1,1}=w_{1,k+1}={w_k\over k}\ . $$
The preferential attachment model is not covered by our Lemma~\ref{perronfrob} (since condition (2) is not satisfied). 
But the preferential attachement model, in particular the vertex degree distribution $\rho_k$ can be studied by different methods (see e.g. \cite{Svante}).
It can also be recovered as a limiting case of our model, by letting the $w_{i,j}\to 0$ if $i$ or $j>1$ in our solution for the $\rho_k$, and eventually taking also as maximal degree $d=d_{\mathrm{max}}=\infty$.
In this limit we also recover for the Hausdorff dimension
$$d_H\to\infty$$
as expected.

\medskip
Secondly, it also contains as a special limit case the so called alpha-model introduced by D. Ford \cite{alphamodel}, in connection with phylogenetic trees and cladograms. This model is related (but is in general not equivalent) to the so called beta-models introduced by Aldous
\cite{AldousBeta,AldousClad}. The alpha-model has been studied in \cite{HassMiermontPitmanWinkel,PitmanWinkel} and generalised into the alpha-gamma model in \cite{ChenFordWinkel} (a mixture of the alpha-model and of the preferential attachment model). It is beyond the scope of this paper to discuss in detail the relationship between these different models, as well as the motivations for them. However, we discuss briefly the basic idea.

The alpha model is a growth model for rooted binary trees, with only 1--vertices (the leaves) and 3--vertices (internal vertices), so that $n_1=n_3+1$ (excluding the root). The growth process grafts a new leaf to a given edge, with probability weight $\alpha$ if the edge is internal (connects two internal vertices) and weight $1-\alpha$ if the edge is external (connects a leaf to an internal vertex).
Thus at each step a leaf and an internal vertex are created.

This leaf-edge attachment process can be easily reproduced in our vertex-splitting model by considering the general model with $d=3$ (we allow only 1--, 2-- and 3--vertices) and by taking the weight $w_{3,1}\to +\infty$. Indeed, in this case, if we start from a binary tree $T$ with no 2--vertices, any splitting will produce a tree $T'$ with a 2--vertex, either by splitting a leaf into a 1--vertex and a 2--vertex (with probability weight $w_{2,1}$), or by splitting an internal vertex into a 3--vertex and a 2--vertex (with probability weight $w_{3,2}$).
But since $w_{3,1}=+\infty$, at the next step, with probability $1$, the new 2--vertex will split into a 1--vertex and a 3--vertex, thus producing a new binary tree where a new 1--vertex has been grafted on an edge. 
These two  two-step processes are depicted in Fig.~\ref{2stepsplit}.
\begin{figure}[!h]
\begin{center}
\includegraphics[width=4in]{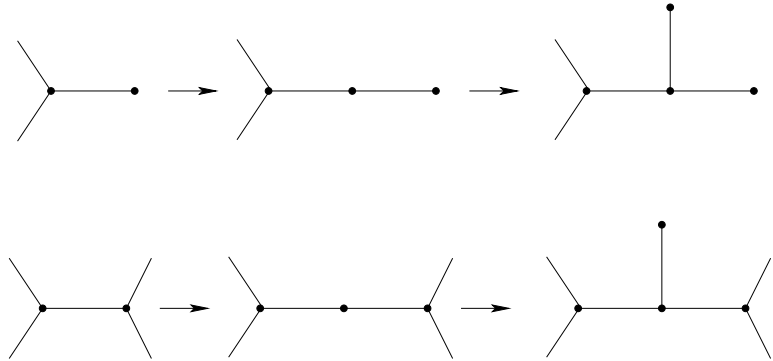}
\caption{The two step processes in the $d=3$ model which reproduce the grafting processes of the $\alpha$-model.}
\label{2stepsplit}
\end{center}
\end{figure}
For large binary trees, the density of 1--vertices and 3--vertices are equal ($\rho_1=\rho_3=1/2$, $\rho_2=0$), as well as the density of internal and external edges ($\rho_{3,1}=\rho_{3,3}=1/2$).
It is easy to see that our model for $d=3$ and $w_{3,1}=+\infty$ should be equivalent to the $\alpha$-model with the choice
\begin{equation}
\label{alphaw32w21}
\alpha={2 w_{3,2}\over w_{2,1}+3\,w_{3,2}}
\quad.
\end{equation}
As expected, $w_{2,2}$ is irrelevant in the limit $w_{3,1}\to+\infty$.
This equivalence can be checked explicitly. For instance it is proven in \cite{alphamodel} that the so-called Sackin's index $S(t)$ scales with the size $t$ of the binary tree as $S(t)\sim t^{1+\alpha}$. Here $S(t)$ is defined as the sum over the leaves $v$ of the number of edges between $v$ and the root (minus 1). It is therefore related to the radius of the tree, and we expect it to scale as $t^{1+1/d_H}$, where $d_H$ is the Hausdorff dimension of the tree. Therefore we expect that 
\begin{equation}
\label{alphadH}
\alpha=1/d_H
\end{equation}
in the $\alpha$-model.
This is indeed the case, taking the limit $w_{3,1}\to +\infty$ in the general formula (\ref{dHd=3}) for $d_H$ in the $d=3$ model, we obtain
\begin{equation}
\label{limd_H}
\lim_{w_{3,1}\to\infty} d_H={w_{2,1}+3w_{3,2}\over 2\,w_{3,2}}
\quad.
\end{equation}

The alpha-gamma model of \cite{ChenFordWinkel} is a model of trees which encompasses both growth by grafting of leaves on existing branches (like in the alpha model) and grafting leaves on existing nodes (like in the attachement models). It is very easy to see that the model presented here contains also  the alpha-gamma model of \cite{ChenFordWinkel}, but we shall not elaborate further on this point.

\medskip
Another class of models are the tree growth models considered in \cite{francoisetal}, which are equivalent to arch deposition models studied in connection with random bond RNA folding problems (in the greedy approximation scheme). These last models are slightly more complex, since they involve a two step splitting process, and their dynamics (i.e. the splitting weights) is fixed by the geometry of RNA folding (the splitting weights cannot be choosen arbitrarily). Variants of these models exhibit universal properties: the fractal dimension is always found to be $d_H=(\sqrt{17}+3)/4$ and some scaling functions are universal. It would be interesting to study correlations for these models with the methods presented here.

Both classes of models (the alpha and alpha-gamma models and the RNA models) are related to fragmentation trees and to coalescent trees, that is trees which appear in continuous time processes. It would also be interesting to study our model by continuous time methods. 
Another important question is universality: we have seen that our model, like the models of \cite{alphamodel} and \cite{AldousBeta,AldousClad}  gives trees with continuously varying scaling exponents. In particular we can obtain $d_H=1$, $d_H=2$ and $d_H=\infty$. We do not know, however, whether the trees that we obtain in this case have the same scaling limit (scaling exponents and correlation functions) as the trees obtained in \cite{alphamodel} for $\alpha=1$, $1/2$ an $0$ respectively (the Comb, the Uniform and the Yule models of phylogenetic trees) or if our model gives a much richer universality class structure.
Finally the trees studied here are symmetric, in the sense that the ``time ordering'' of the vertices, or their distance from the root, does not seem to be crucial, since the growth rule by splitting is local and geometric. It should be interesting to understand the similarities and the difference with the directed trees considered in (for instance) directed polymers and evolving population problems (see e.g. \cite{BrunetDerridaSimon2008} and reference therein).

\subsection*{Acknowledgements}
This research is supported in part by the French Agence Nationale de la Recherche
program GIMP ANR-05-BLAN-0029-01 (F.D.), the ENRAGE European network,
MRTN-CT-2004-5616, the Programme d'Action Int\'egr\'ee J. Verne ``Physical applications of random graph theory'',
the University of Iceland Research Fund and the Eimskip Research Fund at the University of Iceland.
F.D. thanks G. Miermont and N. Curien for their interest and useful discussions.

\appendix 
\section{Illustration of recursion equations for correlations} 
\label{A:corr}
To make the notation more compact we will write for $i\leq j$
$$ \ell_{i,j} = \ell_i,\ldots,\ell_j, \quad \text{and} \quad |\ell_{i,j}| = \ell_i+\ldots+\ell_j.$$
We can write the following recursions for going from time $\ell-1$ to 
time $\ell$. Note that $s<\ell$ in (\ref{e_r6}), (\ref{e_r8}) and 
(\ref{e_r9}).
{\small{
\begin{eqnarray} 
\lefteqn{p_{1k}(\ell''_{1,k-1};s) \;\;=  } \nonumber \\
&& \frac{1}{W(\ell-1)+w_1}\Big(\sum_{i=1}^{k-1} W(\ell''_i-1)
	p_{1k}(\ell''_{1,i-1},\ell''_i-1,\ell''_{i+1,k-1};s) \nonumber
	\\  && \hskip2in + \;\; \delta_{k2}w_1p_R(\ell-1;s) 
	+ \delta_{\ell1}\delta_{k1}w_1\Big). \nonumber\\
	\label{e_r6} &&\\
\lefteqn{p_{1k}(\ell''_{1,k-1};\ell)\;\;=\;\;}\nonumber\\
&& \dfrac{1}{W(\ell-1)+w_1}\sum_{s=0}^{\ell-1}
	\sum_{i=k-1}^d(k-1)w_{k,i+2-k} 
	\sum_{|\tilde{\ell}_{1,i+1-k}|=\ell''_1-1}
p_{1i}(\tilde{\ell}_{1,i+1-k},\ell''_{1,k-1};s).\nonumber\\ \label{e_r7}
\end{eqnarray}
}}

{\small{
\begin{eqnarray} 
\lefteqn{p_{j1}(\ell'_{1,j-1};s) \;\;=\;\; } \nonumber \\
&& \dfrac{1}{W(\ell-1)+w_1}\Big(\sum_{i=1}^{j-1}W(\ell'_i-1)
p_{j1}(\ell'_{1,i-1},\ell'_{i}-1,\ell'_{i+1,j-1};s) 
 \\ && \hskip1in + \;\; (j-1)w_{j,1}p_{j-1}(\ell'_{1,j-1};s)\nonumber\\
&& + \sum_{i=j-1}^d \dfrac{2 w_{j,i+2-j}}{i-1}\sum_{p=1}^{j-2}
	\sum_{n=p+1}^{j-1} 
	\sum_{|\tilde{\ell}_{1,i+1-j}|=\ell'_n-1}
	p_{i1}(\ell'_{1,n-1},\tilde{\ell}_{1,i+1-j},\ell'_{n+1,j-1};s),\nonumber\\
&& + \;\; (j-1)\sum_{i=j}^{d}\dfrac{(i-j+1) w_{j,i+2-j}}{i-1}
  \sum_{|\tilde{\ell}_{1,i+1-j}|=\ell'_1-1}
  p_{i1}(\tilde{\ell}_{1,i+1-j},\ell'_{2,j-1};s)\Big),\nonumber \\ 
	\label{e_r8}&&  \\
	\lefteqn{p_{j1}(\ell'_{1,j-1};\ell) \;\;=\;\; 0.} 
		\nonumber
\end {eqnarray}
}}
\begin{figure}[!h]
\centerline{\scalebox{0.7}{\includegraphics{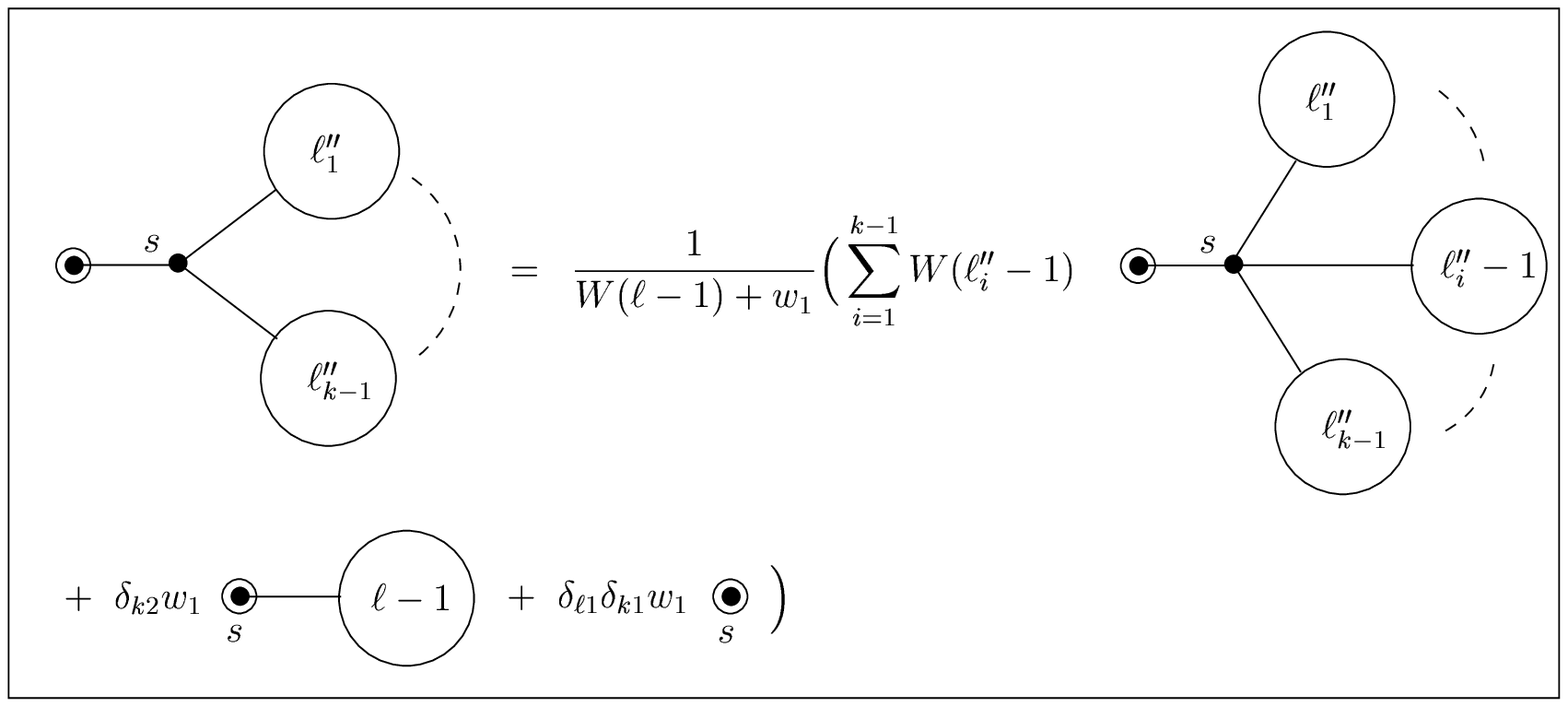}}} 
\caption{
Illustration of equation (\ref{e_r6}).} \label{f_r6}
\end{figure}
\begin{figure}[!h]
\centerline{\scalebox{0.7}{\includegraphics{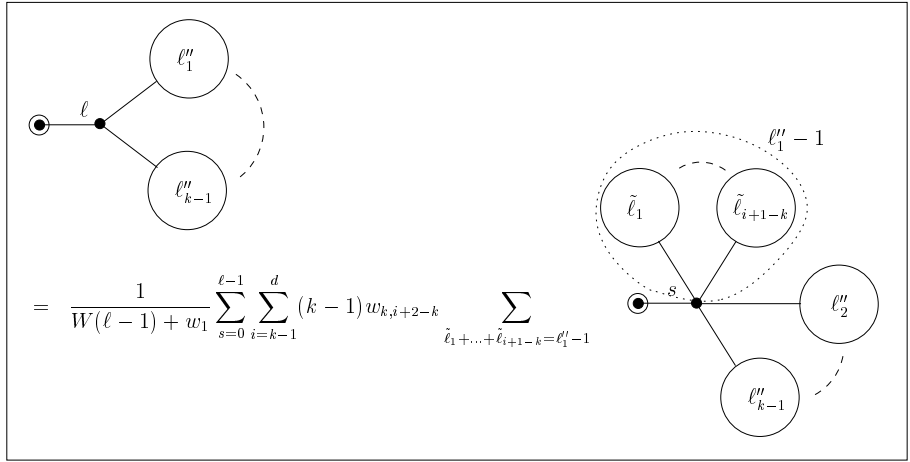}}} 
\caption{
Illustration of equation (\ref{e_r7}).} \label{f_r7}
\end{figure}

In deriving equations (\ref{e_r8}) and (\ref{e_r9}) and the corresponding 
figures, note that the index $p$ is introduced in the second last 
diagram in each figure. The reason for it is the following: 
even though the functions $p_{j1}(\ell'_1,\ldots,\ell'_{j-1};s)$ and 
$p_{jk}(\ell'_1,\ldots,\ell'_{j-1};\ell''_1,\ldots,\ell''_{k-1};s)$ 
are symmetric under permutations of $(\ell'_2,\ldots,\ell'_{j-1})$ it 
does matter where the edge going from $s$ towards the root, is located. 

Therefore, we group together the balloons counter-clockwise from $s$ towards 
the rooted balloon and the balloons clockwise from $s$ 
towards the rooted balloon into another group, one of the groups is
possibly empty. 
If the total number of balloons in the groups is $i-2$ then there are 
$i-1$ such possible configurations. 
In the equations we therefore divide by $i-1$ and sum over all the 
configurations which contribute to the configuration on the 
left of the equality sign. 

The index $p$ in the sum is the location of $s$ counter-clockwise from 
the rooted balloon. Note that $p$ can be no larger than $j-2$ since if 
it were larger, there would be no space for the balloons inside the 
dotted circle. 
Note that the balloons inside the dotted circle are always drawn 
clockwise from the vertex $s$. 
To count the possibility that they are counter-clockwise from $s$ we 
multiply by 2.
\begin{figure}[!h]
\centerline{\scalebox{0.7}{\includegraphics{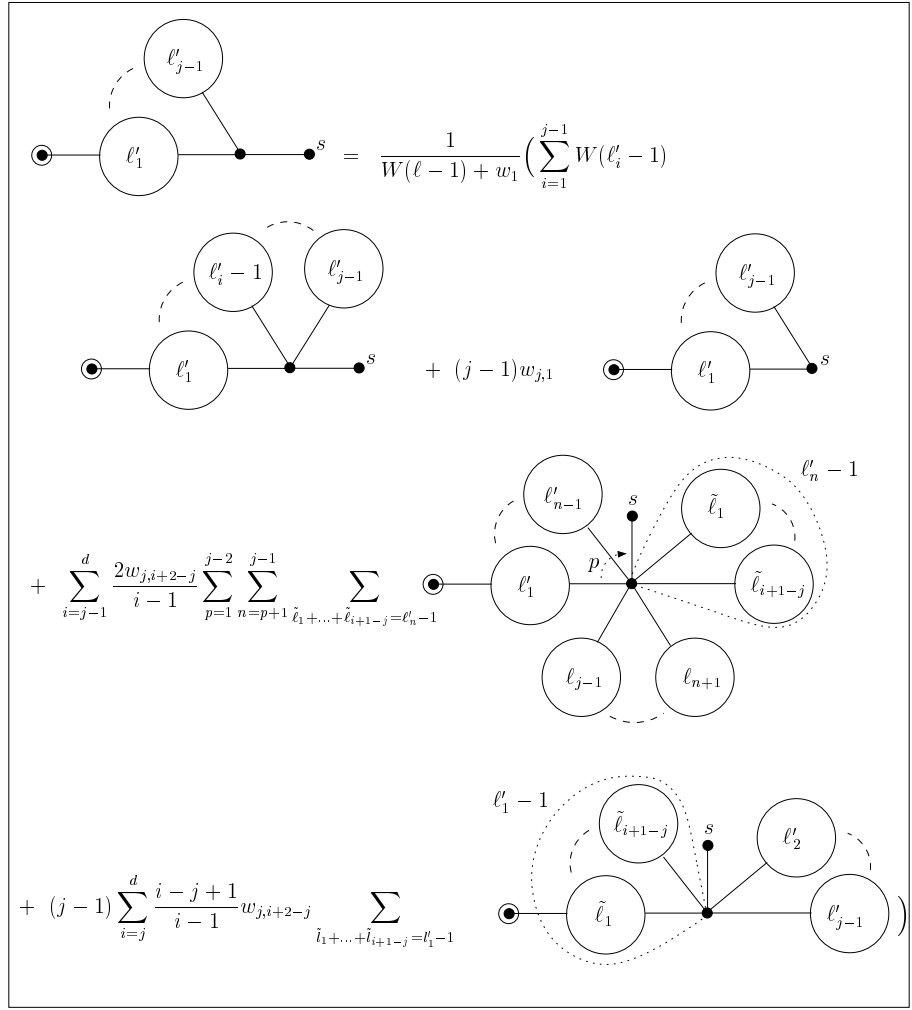}}} 
\caption{
Illustration of equation (\ref{e_r8}).} \label{f_r8}
\end{figure}
{\small{
\begin {eqnarray} \nonumber	
\lefteqn{p_{jk}(\ell'_{1,j-1};\ell''_{1,k-1};s) \;\;=} \nonumber \\
&& \frac{1}{W(\ell-1)+w_1} \Big(\sum_{i=1}^{j-1}W(\ell'_i-1)
	p_{jk}(\ell'_{1,i-1},\ell'_i-1,\ell'_{i+1,j-1};\ell''_{1,k-1};s) \nonumber \\ 
&&+\;\; \sum_{i=1}^{k-1}W(\ell''_i-1)p_{jk}(\ell'_{1,j-1};\ell''_{1,i-1},\ell''_i-1,\ell''_{i+1,k-1};s)  
\nonumber  \\ && + \;\; (j-1)w_{j,k}p_{j+k-2}(\ell'_{1,j-1},\ell''_{1,k-1};s)  \nonumber\\
&&+ \sum_{i=j-1}^{d}\frac{2 w_{j,i+2-j}}{i-1}\sum_{p=1}^{j-2}
	\sum_{n=p+1}^{j-1}\sum_{|\tilde{\ell}_{1,i+1-j}|=\ell'_n-1} p_{ik}(\ell'_{1,n-1},\tilde{\ell}_{1,i+1-j},\ell'_{n+1,j-1};\ell''_{1,k-1};s)  \nonumber\\
&&+ \;\;(j-1)\sum_{i=j}^{d}\frac{i-j+1}{i-1}w_{j,i+2-j}
	\sum_{|\tilde{\ell}_{1,i+1-j}|=\ell'_1-1} 
 p_{ik}(\tilde{\ell}_{1,i+1-j},\ell'_{2,j-1};\ell''_{1,k-1};s)\Big), \nonumber \\ \label{e_r9} 
\end{eqnarray}
}}

\begin{figure}[!h]
\centerline{\scalebox{0.7}{\includegraphics{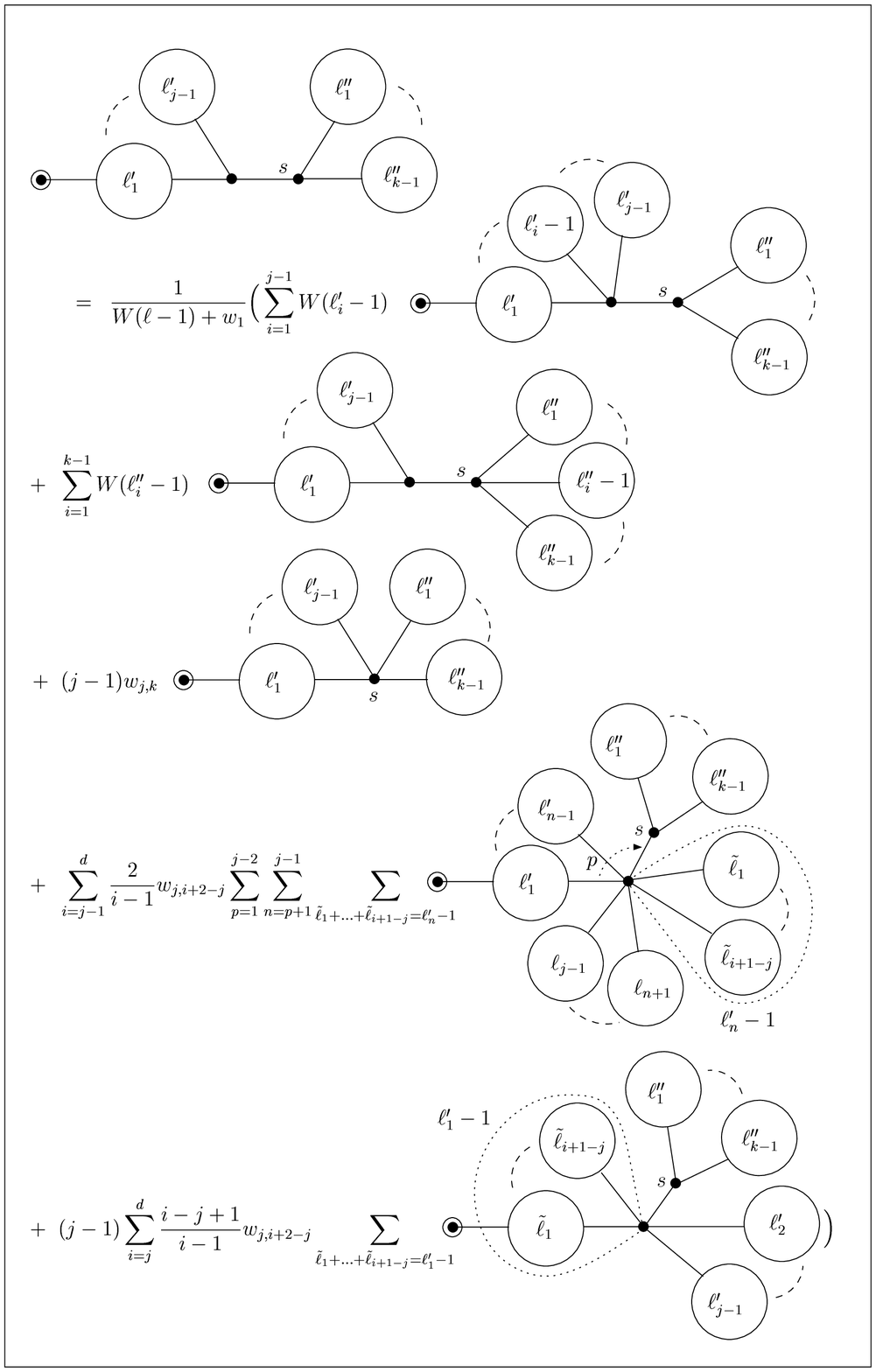}}} 
\caption{
Illustration of equation (\ref{e_r9}).} \label{f_r9}
\end{figure}

{\small{
\begin{eqnarray} 
\lefteqn{p_{jk}(\ell'_{1,j-1};\ell''_{1,k-1};\ell) \;\;=\;\; \dfrac{1}{W(\ell-1)+w_1}\times}\nonumber\\
&& \sum_{s=0}^{\ell-1}\sum_{n=1}^{k-1}\sum_{i=k-1}^{d}w_{k,i+2-k} 
\sum_{|\tilde{\ell}_{1,i+1-k}|=\ell''_n-1} 
p_{ji}(\ell'_{1,j-1};\ell''_{1,n-1},
	\tilde{\ell}_{1,i+1-k},\ell''_{n+1,k-1};s).\nonumber \\ \label{e_r10}
\end{eqnarray}
}}
\begin{figure}[!h]
\centerline{\scalebox{0.7}{\includegraphics{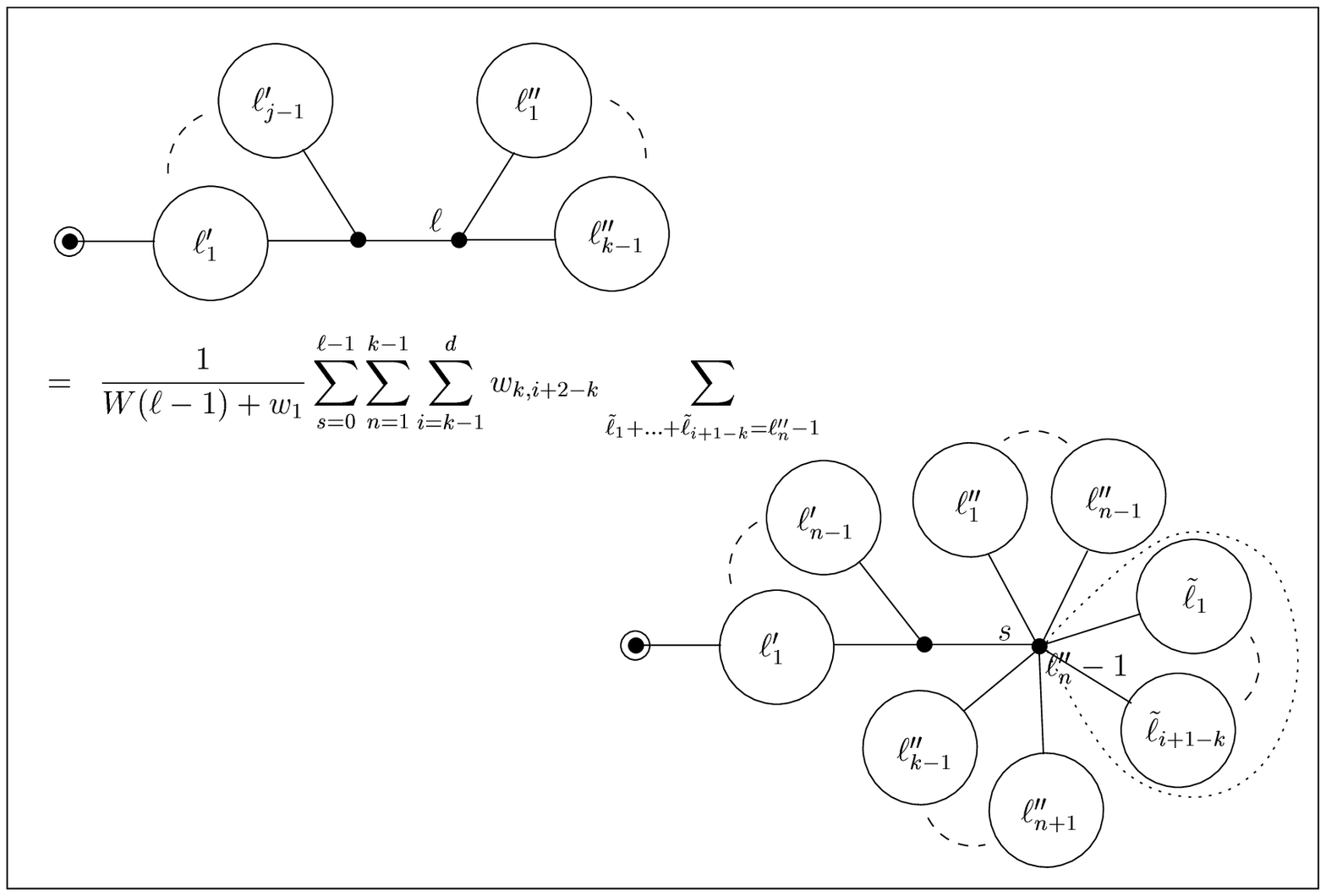}}} 
\caption{Illustration of equation (\ref{e_r10}).} 
\label{f_r10}
\end{figure}

Now average over the label $s$ as before and get the following 
recursion, going from time $\ell$ to $\ell+1$
{\small{
\begin {eqnarray*}
\lefteqn{p_{1k}(\ell''_{1,k-1}) \;\;=}\nonumber \\
&&  \dfrac{\ell+1}{\ell+2}\dfrac{1}{W(\ell)+w_1} \Big(\sum_{i=1}^{k-1} 
W(\ell''_i-1)p_{1k}(\ell''_{1,i-1},\ell''_i-1,\ell''_{i+1,k-1})  
+ \delta_{k2}w_1p_R(\ell) \nonumber \\ && + \;\;\delta_{\ell,0}\delta_{k1}w_1 
+\;\;  (k-1)\sum_{i=k-1}^dw_{k,i+2-k} 
\sum_{|\tilde{\ell}_{1,i+1-k}|=\ell''_1-1}
p_{1i}(\tilde{\ell}_{1,i+1-k},\ell''_{1,k-1}) \Big).
\end{eqnarray*}
}}
{\small{
\begin{eqnarray*}
\lefteqn{p_{j1}(\ell'_{1,j-1}) \;\;=} \\
&& \dfrac{\ell+1}{\ell+2}\dfrac{1}{W(\ell)+w_1} 
\Big(\sum_{i=1}^{j-1}W(\ell'_i-1)
p_{j1}(\ell'_{1,i-1},\ell'_{i}-1,\ell'_{i+1,j-1}) 
 \\ && \hskip1.1in+\;\; (j-1)w_{j,1}p_{j-1}(\ell'_{1,j-1})\\
&& + \sum_{i=j-1}^d \frac{2w_{j,i+2-j}}{i-1}\sum_{p=1}^{j-2}
\sum_{n=p+1}^{j-1} 
\sum_{|\tilde{\ell}_{1,i+1-j}|=\ell'_n-1}
p_{i1}(\ell'_{1,n-1},\tilde{\ell}_{1,i+1-j},\ell'_{n+1,j-1}) \\ 
&& +\;\;(j-1)\sum_{i=j}^{d}\frac{i-j+1}{i-1}w_{j,i+2-j} 
\sum_{|\tilde{\ell}_{1,i+1-j}|=\ell'_1-1}
p_{i1}(\tilde{\ell}_{1,i+1-j},\ell'_{2,j-1})\Big).
\end{eqnarray*}
}}

{\small 
\begin{eqnarray*}
\lefteqn{
 p_{jk}(\ell'_{1,j-1};\ell''_{1,k-1}) \;\;=
	} \\
&&  \dfrac{\ell+1}{\ell+2}\dfrac{1}{W(\ell)+w_1} 
\Big(\sum_{i=1}^{j-1}W(\ell'_i-1)
p_{jk}(\ell'_{1,i-1},\ell'_i-1,\ell'_{i+1,j-1};\ell''_{1,k-1}) \\
&& +\;\; \sum_{i=1}^{k-1}W(\ell''_i-1)
p_{jk}(\ell'_{1,j-1};\ell''_{1,i-1},\ell''_i-1,\ell''_{i+1,k-1}) 
 \nonumber \\ && +\;\; (j-1)w_{j,k}p_{j+k-2}(\ell'_{1,j-1},\ell''_{1,k-1}) \\
&&+ \sum_{i=j-1}^{d}\frac{2w_{j,i+2-j}}{i-1}\sum_{p=1}^{j-2}
\sum_{n=p+1}^{j-1}
\sum_{|\tilde{\ell}_{1,i+1-j}|=\ell'_n-1}
 p_{ik}(\ell'_{1,n-1},\tilde{\ell}_{1,i+1-j},
	\ell'_{n+1,j-1};\ell''_{1,k-1}) \\ 
&& +\;\; (j-1)\sum_{i=j}^{d}\frac{i-j+1}{i-1}w_{j,i+2-j}
\sum_{|\tilde{\ell}_{1,i+1-j}|=\ell'_1-1} 
p_{ik}(\tilde{\ell}_{1,i+1-j},\ell'_{2,j-1};\ell''_{1,k-1})\\
&& + \;\;\sum_{n=1}^{k-1}\sum_{i=k-1}^{d}w_{k,i+2-k} 
\sum_{|\tilde{\ell}_{1,i+1-k}| =\ell''_n-1} 
p_{ji}(\ell'_{1,j-1};\ell''_{1,n-1},
	\tilde{\ell}_{1,i+1-k},\ell''_{n+1,k-1})\Big).
\end{eqnarray*}}

\end{document}